\documentclass[12pt,a4paper]{article}
\usepackage[latin1]{inputenc}
\usepackage{amsmath}
\usepackage{amsthm}
\usepackage{amsfonts}
\usepackage{vmargin}
\usepackage{amssymb}
\usepackage{stmaryrd}
\author{Christopher Shirley\thanks{The author thanks his supervisor Frédéric Klopp, for his advice and guidance over the course of the study}}
\title{Decorrelation estimates for some continuous and discrete random schr\"{o}dinger operators in
dimension one and applications to spectral statistics}
\newtheorem{theo}{Theorem}[section]

\newtheorem{prop}[theo]{Proposition}
\newtheorem{lem}[theo]{Lemma}

\newtheorem{rem}[theo]{Remark}
\numberwithin{equation}{section}
 
\newcommand{\R}{\mathbb{R}}
\newcommand{\N}{\mathbb{N}}
\newcommand{\Z}{\mathbb{Z}}
\newcommand{\C}{\mathbb{C}}

\begin{document}
\maketitle

\begin{abstract}
The purpose of the present work is to establish decorrelation estimates at distinct energies for some random Schrödinger operator in dimension one. In particular, we establish the result for some random operators on the continuum with alloy-type potential. These results are used to give a description of the spectral statistics.
\end{abstract}

\section{Introduction}
To  introduce our results, let us first consider one of the random operators that will be studied in the rest of this article. Let $(\omega_n)_{n\in\Z}$ be independent random variables, uniformly distributed on $[0,1]$ and define the random potential $V_\omega(x)=\omega_n$ on $(n,n+1)$. Consider the operator $H_\omega: L^2(\R)\rightarrow  L^2(\R)$ defined by the following equation

\begin{equation}\label{simpleop}
\forall \phi\in \mathcal{H}^2(\R), H_\omega \phi=-\Delta \phi+V_\omega \phi.
\end{equation}

We know that, with probability one, $H_\omega$ is self-adjoint. As $H_\omega$ is $\Z$-ergodic, we know that there exists a set $\Sigma$ such that, with probability one, the spectrum of $H_\omega$ is equal to $\Sigma$ (see for instance \cite{CL90}). One of the purposes of this article is to give a description of the spectral statistics of $H_\omega$. In this context, we study the restriction of $H_\omega$ to a finite box and study the diverse statistics when the size of the box tends to infinity. For $L\in\N$, let $\Lambda_L=[-L,L]$ and $H_\omega(\Lambda_L)$ be the restriction of $H_\omega$ to $L^2(\Lambda_L)$ with Dirichlet boundary conditions. The spectrum of $H_\omega(\Lambda_L)$ is discrete, accumulate at $+\infty$. We denote $(E_j)_{j\in\N}$ the eigenvalues of $H_\omega(\Lambda)$, ordered increasingly and repeated according to multiplicity. 
We know from the $\Z$-ergodicity that there exists a deterministic, nondecreasing function $N$ such that, almost surely, we have
\begin{equation}
N(E)=\lim_{L\to\infty}\dfrac{\sharp\{j,E_j<E\}}{|\Lambda_L|}.
\end{equation}
The function $N$ is the integrated density of state (abbreviated IDS from now on), and it is the distribution function of a measure $dN$.

In order to study the spectral statistics of $H_\omega(\Lambda)$ we use four results : the localization assumption, the Wegner estimates, the Minami estimates and the decorrelation estimates for distinct energies. They will be introduced in the rest of the section.

Let $\mathcal{I}$ be an open relatively compact subset of $\R$. We know from \cite{K14} that the operator satisfies the following localization assumption.

\textbf{(Loc): } for all $\xi\in(0,1)$, one has
\begin{equation}
\sup_{L>0} 
\underset{|f|\leq 1}{\underset{\text{supp }f\subset \mathcal{I}}\sup}
\mathbb{E}\left(\sum_{\gamma\in\Z^d}e^{|\gamma|^\xi}\|\textbf{1}_{[-1/2,1/2]}f(H_\omega(\Lambda_L))\textbf{1}_{[\gamma-1/2,\gamma+1/2]}\|_2\right)<\infty
\end{equation}
We know (see for instance \cite{CHK07}) that the following Wegner estimates hold on $\mathcal{I}$: 

\textbf{(W) :} There exists $C>0$,   such that for $J\subset \mathcal{I}$ and $L\in\N$ 
\begin{equation}
\mathbb{P}\Big[\text{tr} \left(\textbf{1}_J(H_\omega(\Lambda_L)) \right)\geq 1\Big ]\leq C |J||\Lambda_L|.
\end{equation}
This shows that the integrated density of state (abbreviated IDS from now on) $N(.)$ is Lipschitz continuous. As the IDS is a non-decreasing function, this implies that $N$ is almost everywhere differentiable and its derivative $\nu(.)$ is positive almost-everywhere on its essential support.

Let us now introduce the Minami estimates. We extract from \cite{K11} the
\begin{theo}[M]\label{mina-int}
Fix $J\subset\mathcal{I}$ a compact interval. For any $s'\in(0,1)$, $M>1$, $\eta>1$, $\rho\in(0,1)$, there exists $L_{s',M,\eta,\rho}>0$ and $C=C_{s',M,\eta,\rho}>0$ such that, for $E\in J$, $L\geq L_{s',M,\eta,\rho}$ and $\epsilon\in[L^{-1/s'}/M,ML^{-1/s'}]$ , one has
\begin{displaymath}
\sum_{k\geq2}\mathbb{P}\big(\textup{tr}[\textbf{1}_{[E-\epsilon,E+\epsilon]}(H_\omega(\Lambda_L))]\geq k\big)\leq C( \epsilon L )^{1+\rho}.
\end{displaymath}
\end{theo}

One purpose of this article is, as in \cite{GK10}, to give a description of spectral statistics. For instance, we obtain the following result. Define the \textit{unfolded local level statistics} near $E_0$ as the following point process :
\begin{equation}
\Xi(\xi;E_0,\omega,\Lambda)=\sum_{j\geq1} \delta_{\xi_j(E_0,\omega,\Lambda)}(\xi)
\end{equation}
 where
 \begin{equation}
 \xi_j(E_0,\omega,\Lambda)=|\Lambda|(N(E_j(\omega,\Lambda)-N(E_0)).
 \end{equation}
The unfolded local level statistics are described by the following theorem which corresponds to \cite[Theorem 1.9]{GK10} with a stronger hypothesis.

\begin{theo}\label{ULLS}
Pick $E_0\in \mathcal{I}$ such that $N(.)$ is differentiable at $E_0$ and $\nu(E_0)>0$.Then, when $|\Lambda|\to \infty$, the point process
$\Xi(\xi;E_0,\omega,\Lambda)$ converges weakly to a Poisson process with intensity the Lebesgue measure. That is, for any $p\in\N^*$, for any $(I_i)_{i\in\{1,\dots,p\}}$ collection of disjoint intervals
\begin{equation}
\lim_{|\Lambda|\to\infty}\mathbb{P}
\left(
\left\{\omega;
\begin{aligned}
\sharp\{j;\xi_j(\omega,\Lambda)\in I_1\}=k_1\\
\vdots\hspace*{8em} \vdots\hspace*{1em}\\
\sharp\{j;\xi_j(\omega,\Lambda)\in I_p\}=k_p
\end{aligned}
\right\}\right)=\dfrac{|I_1|^{k_1}}{k_1!}\dots\dfrac{|I_p|^{k_p}}{k_p!}
\end{equation}

\end{theo}

Now, one can wonder what is the joint behaviour at large scale of the point processes $\Xi(\xi;E_0,\omega,\Lambda)$ and $\Xi(\xi;E_1,\omega,\Lambda)$ with $E_0\neq E_1$. We obtain the following theorem which corresponds to \cite[Theorem 1.11]{GK10}.

\begin{theo}\label{joint}
Pick $(E_0,E_0')\in \mathcal{I}^2$ such that $E_0\neq E_0'$ and such that $N(.)$ is differentiable at $E_0$ and $E_0'$ with $\nu(E_0)>0$ and $\nu(E_0')>0$.\\
When $|\Lambda|\rightarrow \infty$ the point processes $ \Xi(E_0,\omega,\Lambda)$ and $\Xi(E_0',\omega,\Lambda)$, converge weakly respectively to two independent Poisson processes on $\R$ with intensity the Lebesgue measure. That is, for any $(J_+,J_-)\in(\N^*)^2$, for any $(U_j^+)_{1\leq j\leq J_+}\subset\R^{J_+}$ and $(U_j^-)_{1\leq j\leq J_-}\subset\R^{J_-}$ collections of disjoint compact intervals, one has
\begin{displaymath}
\mathbb{P}\left(
\begin{aligned}
\sharp\{j;\xi_j(E_0,\omega,\Lambda)\in U_1^+\}&=k_1^+\\
\vdots\hspace*{8em} &\vdots\\
\sharp\{j;\xi_j(E_0,\omega,\Lambda)\in U_{J_+}^+\}&=k_{J_+}^+\\
\sharp\{j;\xi_j(E_0',\omega,\Lambda)\in U_1^-\}&=k_1^-\\
\vdots\hspace*{8em} &\vdots\\
\sharp\{j;\xi_j(E_0',\omega,\Lambda)\in U_{J_-}^-\}&=k_{J_-}^-
\end{aligned}
\right)\underset{|\Lambda|\to\infty}{\rightarrow} \prod_{j=1}^{J_+}\dfrac{|U_j^+|^{k_j^+}}{k_j^+!}e^{-|U_j^+|}\cdot\prod_{j=1}^{J_-}\dfrac{|U_j^-|^{k_j^-}}{k_j^-!}e^{-|U_j^-|}.
\end{displaymath}
\end{theo}

To prove this theorem we use decorrelation estimates at distinct energies. 
\begin{theo}\label{decoInt}
There exists $\gamma>0$ such that for any $\beta\in(1/2,1)$, $\alpha\in (0,1) $ and $(E,E')\in(\R)^2$ such that at $E\neq E'$, for any $k>0$ there exists $C>0$ such that for $L$ sufficiently large and $kL^\alpha\leq l\leq L^\alpha/k$ we have 
\begin{displaymath}
\mathbb{P}\left(
\begin{aligned}
 \textup{tr}\, \textbf{1}_{[E-L^{-1},E+L^{-1}]}\left(H_\omega(\Lambda_l)\right)\neq 0,\\
 \textup{tr}\, \textbf{1}_{[E'-L^{-1},E'+L^{-1}]}\left(H_\omega(\Lambda_l)\right)\neq 0
\end{aligned}
\right)\leq C\dfrac{l^2}{L^{1+\gamma}}.
\end{displaymath}
\end{theo}
As Theorem~\ref{mina-int} is used to prove Theorem~\ref{ULLS}, Theorem~\ref{decoInt} is used to prove Theorem~\ref{joint} (see \cite{GK10}).
\section{Models and Main result}
In this section, we introduce the models that will be studied,  the main result of this article and the assumptions made to prove this result. Let $(\omega_n)_{n\in\Z}$ be independent random variables with a common bounded, compactly supported density $\mu$.
\subsection{Models}
In this article, we study models on the continuum but also generalize the results of \cite{S13,K11,T14} to other discrete models. They will be introduce in the two following paragraphs. 
\paragraph{Alloy-type models}
\begin{itemize}
\item[$\bullet$] We first introduce continuous alloy-type models. Fix $q:\R\to\R$ a single-site potential that satisfies the following hypotheses : \\
\textbf{(H1) :} $q$ is piecewise continuous.\\
\textbf{(H2) :} There exist $\eta>0$ and $N\in\N^*$ such that
\begin{equation}\label{stepfunction}
\dfrac{1}{\eta}\text{1}_{[-N,N]}\geq q\geq \eta\text{1}_{[-1/2,1/2]} .
\end{equation}
The right-hand side of \eqref{stepfunction} is often called "covering condition", we do not know how to relax this hypothesis. Define the random potential 
\begin{equation}
V_\omega(x)=\sum_{n\in\Z} \omega_n\, q(x-n)=\sum_{n\in [x-N,x+N]\cap\Z} \omega_n\, q(x-n)
\end{equation}
We consider the random operator $H_\omega: L^2(\R)\rightarrow  L^2(\R)$ defined by the following equation : 
\begin{equation}\label{fullopcont}
\forall \phi\in \mathcal{H}^2(\R), H_\omega \phi=-\Delta \phi+V_\omega \phi.
\end{equation}
As the the proof will show, it is possible to weaken the right-hand side of \eqref{stepfunction} and suppose that $q$ is non-negative and that the set $q^{-1}\{0\}\cap[-1/2,1/2]$ has no accumulation points (see Remark~\ref{remzero}).\\
\item[$\bullet$] We now introduce discrete alloy-type models. Fix $(d_n)_{n\in\Z}\in\R^\Z$ a non-zero discrete single-site potential satisfying the following hypotheses : \\
\textbf{(H1) :} $(d_n)_{n\in\Z}\in\R_+^\Z$ or $(d_n)_{n\in\Z}\in\R_-^\Z$.\\
\textbf{(H2) :} $(d_n)_{n\in\Z}$ is compactly supported. \\
Define the random potential
\begin{equation}
V_\omega(m)=\sum_{n\in\Z} \omega_n\, d_{m-n}.
\end{equation}
We consider the random operator $H_\omega: \ell^2(\Z)\rightarrow  \ell^2(\Z)$ defined by the following equation
\begin{equation}\label{fullopdisalloy}
\forall u\in \ell^2(\Z), H_\omega u=-\Delta u+V_\omega u.
\end{equation}
where $\Delta$ is the discrete Laplace operator.
\end{itemize}
\paragraph{Multimer type models}
We now introduce a class of models that includes the random dimer models. The random dimer models considered in the present article are not the same as the one considered in \cite{DG00}. Indeed, in the present article the random variables has a common density while in \cite{DG00} the random variables are Bernoulli distributed. Let $(b_n)_{n\in\Z}\in\R^\Z$ such that
\begin{equation}
\inf\{|b_n|,n\in\Z\}>0.
\end{equation} Define the jacobi operator $\Delta_b:\ell^2(\Z)\rightarrow\ell^2(\Z)$ by
\begin{equation}\label{defjacmod}
\Delta_b u\,(n)=b_{n+1} u(n+1)+ b_nu(n-1).
\end{equation}
Fix $N\in\N^*$ and $(a_m)_{m\in\llbracket 1,N \rrbracket}\in(\R_+^*)^N$. Define the random potential $V_\omega:\Z\to\R$ by 
\begin{equation}\label{discmod}
\forall j\in\Z,\forall m\in \llbracket 0,N-1\rrbracket, V_\omega(Nj+m)=a_m\omega_j. 
\end{equation}
When all the random variables take the same value, the potential is periodic.
 We consider the random operator $H_\omega: \ell^2(\Z)\rightarrow  \ell^2(\Z)$ defined by the following equation : 
\begin{equation}\label{fullopdis}
\forall \phi\in \ell^2(\Z), H_\omega \phi=-\Delta_b \phi+V_\omega \phi.
\end{equation}
We will suppose that $N\geq 2$. The case $N=1$ is studied in \cite{S13} (see also \cite{K11}). The sequence $(b_n)_n$ is supposed deterministic, but if the sequence is random, independent of the sequence of random variable $(\omega_n)_n$ and if there exists $m>0$ such that $\inf\{|b_n|,n\in\Z\}>m$ almost surely, the decorrelation estimates hold as well.
\subsection{Assumptions}
We suppose there exists a relatively compact, open interval $\mathcal{I}\subset \R$ such that the Wegner estimates hold on $\mathcal{I}$ : 

\textbf{(W) :} There exists $C>0$ such that for $J\subset \mathcal{I}$ and $\Lambda$ an interval in $\R$, one has
\begin{equation}
\mathbb{P}\Big[tr \left(\textbf{1}_J(H_\omega(\Lambda)) \right)\geq 1\Big ]\leq C |J||\Lambda|.
\end{equation}
Wegner estimate has been proven for many different models, discrete or continuous (\cite{K95,CHK07,CGK09,V10}). Assumption (W) implies that the IDS is Lipschitz continuous.
\\
 We suppose that the localization property holds on $\mathcal{I}$ : 

\textbf{(Loc): } for all $\xi\in(0,1)$, one has
\begin{equation}
\sup_{L>0} 
\underset{|f|\leq 1}{\underset{\text{supp }f\subset \mathcal{I}}\sup}
\mathbb{E}\left(\sum_{\gamma\in\Z^d}e^{|\gamma|^\xi}\|\textbf{1}_{\Lambda(0)}f(H_\omega(\Lambda_L))\textbf{1}_{\Lambda(\gamma)}\|_2\right)<\infty
\end{equation}
This property can be shown using either multiscale analysis or fractional moment method. In fact we suppose that $\mathcal{I}$ is a region where we can do the bootstrap multiscale analysis of \cite{GK01}. (Loc) is equivalent to the conclusion of the bootstrap MSA (see \cite[Appendix]{GK10} for details). We do not require estimates on the operator $H_\omega$ but only on $H_\omega(\Lambda_L)$. 

We assume that the following Minami estimates holds on $\mathcal{I}$.\begin{theo}[M]\label{mina-D3}
Fix $J\subset \mathcal{I}$ a compact. For any $s'\in(0,1)$, $M>1$, $\eta>1$, $\rho\in(0,1)$, there exist $L_{s',M,\eta,\rho}>0$ and $C=C_{s',M,\eta,\rho}>0$ such that, for $E\in J,L\geq L_{s',M,\eta,\rho}$ and $\epsilon\in[L^{-1/s'}/M,ML^{-1/s'}]$ , one has
\begin{displaymath}
\sum_{k\geq2}\mathbb{P}\big(\textup{tr}[\textbf{1}_{[E-\epsilon,E+\epsilon]}(H_\omega(\Lambda_L))]\geq k\big)\leq C( \epsilon L )^{1+\rho}.
\end{displaymath}
\end{theo}
It is proven in \cite{K14} that, in dimension one, for the continuum model, if one has independence at a distance and localization, the Minami estimates are an implication of the Wegner estimates. It is proven in \cite{S13} that this statement holds also for discrete models, under the same assumptions. In both cases, the Minami estimates are not as precise as the Minami estimates proven in \cite{CGK09}, but are sufficient for our purpose. For discrete alloy-type models, Minami estimates are also proven in \cite{TV14} but they only hold for single-site potentials whose Fourier transforms do not vanish. Therefore, we will use the Minami estimates proven in \cite{S13} which hold under the assumptions of the present article.
\subsection{Main result}
The purpose of this article is to prove the  
\begin{theo}\label{deco}
There exists $\gamma>0$ such that, for any $\beta\in(1/2,1)$, $\alpha\in (0,1) $, $(F,G)\in\R^2$ such that at $F\neq G$ and $k>0$, there exists $C>0$ such that for $L$ sufficiently large and $kL^\alpha\leq l\leq L^\alpha/k$ we have 
\begin{displaymath}
\mathbb{P}\left(
\begin{aligned}
 \textup{tr}\, \textbf{1}_{[F-L^{-1},F+L^{-1}]}\left(H_\omega(\Lambda_l)\right)\neq 0,\\
 \textup{tr}\, \textbf{1}_{[G-L^{-1},G+L^{-1}]}\left(H_\omega(\Lambda_l)\right)\neq 0
\end{aligned}
\right)\leq C\dfrac{l^2}{L^{1+\gamma}}.
\end{displaymath}
\end{theo}

Decorrelation estimates give more precise results about spectral statistics, such as Theorem~\ref{joint} (see \cite{GK10} for the proof and other results about spectral statistics). They are a consequence of Minami estimates and localization. In \cite{K11}, Klopp proves decorrelation estimates for eigenvalues of the discrete Anderson model in the localized regime. The result is proven at all energies only in dimension one. In \cite{T14}, decorrelation estimates are proven for the one-dimensional tight binding model, i.e when there are correlated diagonal and off-diagonal disorders. In \cite{S13}, decorrelation estimates are also proven for other discrete models, such as Jacobi operators with positive alloy-type potential or the random hopping model, i.e when there is only off-diagonal disorder. We show that this statement also holds for the continuous models defined in \eqref{fullopcont} and discrete models defined in \eqref{fullopdisalloy} and \eqref{fullopdis}. In fact, the proof for alloy-type models will only be given for operators on the continuum. The proof for the discrete alloy-type model is the same as the proof for continuous alloy-type models, using the results of Appendix~\ref{FDEsec} instead of Appendix~\ref{sturm}, and making the obvious modifications due to the discrete structure, as done in Subsection~\ref{subsecdis} for the models defined in \eqref{fullopdis}.

The proof of Theorem~\ref{deco} rely on the study of the gradients of two different eigenvalues. In particular, we show that the probability that they are are co-linear is zero. In \cite{K11}, \cite{T14} and \cite{S13}, this condition could easily be rewritten as a property of eigenvectors. For instance, for the discrete Anderson model, this condition is the system of equations
\begin{equation}
\forall n\in\llbracket -L,L \rrbracket, u^2(n)=v^2(n).
\end{equation}
 where $u$ and $v$ are normalized eigenvector associated to the eigenvalues. These equations can be rewritten easily as $u(n)=\pm v(n)$.
For the continuous model defined in \eqref{simpleop}, the condition of co-linearity is the system of equations
\begin{equation}\label{gradcolisimp}
\forall n\in\llbracket -L,L-1, \rrbracket, \int_n^{n+1}u^2(n)=\int_n^{n+1}v^2(n).
\end{equation}
We show that this system can be rewritten as a system of $2L$ quadratic equations, using basis of solutions on each interval $(n,n+1)$. This system and the fact that the eigenvectors have continuous derivatives will impose conditions on the eigenvectors that are easier to handle.
\section{Proof of Theorem~\ref{deco}}
We follow the proof of \cite[Section 2]{K11}. The only difference is in the proof of Lemma~\ref{probcoli} below which corresponds to \cite[Lemma 2.4]{K11}. The proof of the other intermediate results are the same as in \cite{K11}. Thus, the results will be given without proofs. The proof of Lemma~\ref{probcoli} is the same for discrete and continuous model except for the proof of Lemma~\ref{probcoli}.

Using (M), Theorem~\ref{deco} is a consequence of the following theorem : 

\begin{theo}\label{thdec}
Let $\beta\in(1/2,1)$. For $\alpha\in (0,1) $ and $(F,G)\in\R^2$ with $F\neq G$, for any $k>1$ there exists $C>0$, such that for $L$ large enough and $kL^\alpha\leq l\leq L^\alpha/k$ we have
\begin{displaymath}
\mathbb{P}_0:=\mathbb{P}\left(
\begin{aligned}
  tr\textbf{1}_{[F-2L^{-1},F+2L^{-1}]}(H_\omega(\Lambda_l ))= 1,\\
 tr\textbf{1}_{[G-2L^{-1},G+2L^{-1}]}(H_\omega(\Lambda_l ))= 1
\end{aligned}
\right)\leq C\left(\dfrac{l^2}{L^{4/3}}\right)e^{(\log L)^\beta}.
\end{displaymath}
\end{theo}

We now restrict ourself to the study of the restriction of $H_\omega$ to cubes of size $(\log L)^{1/\xi'}$ instead of length $L^\alpha$. In this context, we extract from \cite[Proposition 2.1]{K11} the 
\begin{prop}\label{(Loc)(I)} : For all $p>0$ and $\xi\in(0,1)$, for L sufficiently large, there exists a set of configuration $\mathcal{U}_{\Lambda_l}$ of probability larger than $1-L^{-p}$ such that if $\phi_{n,\omega}$ is a normalized eigenvector associated to the eigenvalue $E_{n,\omega}\in\mathcal{I}$ and $x_0(\omega)\in \{1,\dots,L\}$ maximize $|\phi_{n,\omega}|$ then 
\begin{equation}\label{expdec}
|\phi_{n,\omega}(x)|\leq L^{p+d} e^{-|x-x_0|^{\xi}}.
\end{equation}
\end{prop}

Now, Theorem~\ref{thdec} is a consequence of the following lemma and Proposition~\ref{(Loc)(I)}. 
\begin{lem}\label{thdec2}
Let $\beta'\in(1/2,1)$. For $\alpha\in (0,1) $ and $(F,G)\in\R^2$ with $F\neq G$, there exists $C>0$ such that for any $\xi'\in(0,\xi)$, $L$ large enough and $\tilde{l}=(\log L)^{1/\xi'}$ we have 
\begin{displaymath}
\mathbb{P}_1:=\mathbb{P} \left( 
\begin{aligned}
 tr\textbf{1}_{[F-2L^{-1},F+2L^{-1}]}(H_\omega(\Lambda_{\tilde{l}} ))= 1,\\
 tr\textbf{1}_{[G-2L^{-1},G+2L^{-1}]}(H_\omega(\Lambda_{\tilde{l}} ))= 1
\end{aligned}
\right)\leq C\left (\dfrac{\tilde{l}^2}{L^{4/3}}\right)e^{\tilde{l}^{\beta'}}.
\end{displaymath}
\end{lem}
The rest of the section is dedicated to the proof of Lemma~\ref{thdec2}. Define $J_L=\left[E-L^{-1},E+L^{-1}\right]$ and $J_L'=\left[E'-L^{-1},E'+L^{-1}\right]$. For $\epsilon\in(2L^{-1},1)$, for some $\kappa>2$, using (M) when the operator $H_\omega(\Lambda_l)$ has two eigenvalues in $[-\epsilon,+\epsilon]$, one has 
\begin{equation}
\mathbb{P}_1\leq C\epsilon^2 l^{\kappa}+\mathbb{P_\epsilon}\leq C\epsilon^2l^2e^{l^\beta}+\mathbb{P}_\epsilon
\end{equation}
where 
\begin{displaymath}
\mathbb{P_\epsilon}=\mathbb{P}(\Omega_0(\epsilon))
\end{displaymath}
and
\[ \Omega_0(\epsilon)= \left\{ \omega;
\begin{aligned}
\sigma(H_\omega(\Lambda_l))\cap J_L&= \{E(\omega)\} \\
\sigma(H_\omega(\Lambda_l))\cap (E-\epsilon,E&+\epsilon)= \{E(\omega)\} \\
\sigma(H_\omega(\Lambda_l))\cap J_L'&= \{E'(\omega)\} \\
\sigma(H_\omega(\Lambda_l))\cap(E'-\epsilon,&E'+\epsilon)= \{E'(\omega)\}
\end{aligned} 
\right \}.
\]
In order to estimate $\mathbb{P}_\epsilon$ we make the following definition. For $(\gamma,\gamma')\in\Lambda_L^2$ let $J_{\gamma,\gamma'}(E(\omega),E'(\omega))$ be the Jacobian of the mapping $(\omega_\gamma,\omega_{\gamma'})\rightarrow (E(\omega),E'(\omega))$ : 
\begin{equation}\label{defjac}
J_{\gamma,\gamma'}(E(\omega),E'(\omega))=\left \vert \begin{pmatrix} \partial_{\omega_\gamma}E(\omega) & \partial_{\omega_{\gamma'}}E(\omega)\\ \partial_{\omega_\gamma}E'(\omega) &\partial_{\omega_{\gamma'}}E'(\omega)
\end{pmatrix} \right \vert
\end{equation} 
and define 
\begin{displaymath}
\Omega^{\gamma,\gamma'}_{0,\beta}(\epsilon)= \Omega_0(\epsilon)\cap \left \{ \omega ;|J_{\gamma,\gamma'}(E(\omega),E'(\omega))|\geq \lambda \right\}.
\end{displaymath} 

When one of the Jacobians is sufficiently large, the eigenvalues depends on  two independent random variables. Thus the probability to stay in a small interval is small. So we divide the proof in two parts, depending on whether all the Jacobians are small. The next lemma shows that if all the Jacobians are small then the gradients of the eigenvalues, which have positive components for the models considered in the present article, must be almost co-linear.
\begin{lem}\label{grad->jac}
Let $(u,v)\in(\R^+)^{2n}$ such that $\|u\|_1=\|v\|_1=1$. Then 
\begin{displaymath}
\max_{j\neq k} \left | \begin{pmatrix} u_j & u_k\\ v_j & v_k \end{pmatrix} \right |^2\geq \dfrac{1}{4n^5}\Vert u-v \Vert_1^2.
\end{displaymath}
\end{lem}
Thus, either one of the Jacobian determinants is not small or the gradient of $E$ and $E'$ are almost co-linear. We now show that the second case happens with a small probability. 
\begin{lem}\label{probcoli}
Let $(F,G)\in\R^2$ with $F\neq G$ and $\beta>1/2$. Let $\mathbb{P}$ denotes the probability that there exist $E_j(\omega)$ and $E_k(\omega)$, simple eigenvalues of $H_\omega(\Lambda_l)$ such that $|F-E_j(\omega)|+|G-E_k(\omega)|\leq e^{-l^\beta}$ and such that
\begin{equation}\label{gradcoli}
\left\|\dfrac{\nabla_\omega\big(E_j(\omega))}{\|\nabla_\omega\big(E_j(\omega))\|}-\dfrac{\nabla_\omega\big(E_k(\omega))}{\|\nabla_\omega\big(E_k(\omega))\|}\right\|\leq e^{-l^\beta}
\end{equation}
then there exists $c>0$ such that
\begin{equation}
\mathbb{P}\leq e^{-c l^{2\beta}}
\end{equation}
\end{lem}
 The proof of this result depends on the model and will be given below in the paper. First, we finish the proof of Lemma~\ref{thdec2}.
 
 Pick $\lambda=e^{-l^\beta}\|\nabla_\omega\big(E_j(\omega))\|\|\nabla_\omega\big(E_k(\omega))\|$. For the models considered in the present article, there exists $C>1$ such that for all $L$, $\|\nabla_\omega\big(E_j(\omega))\|\in[1/C,C]$. This will be proven in the following subsections. Therefore $\lambda\asymp e^{-l^\beta}$. Then, either one of the Jacobian determinant is larger than $\lambda$ or the gradients are almost co-linear. Lemma~\ref{probcoli} shows that the second case happens with a probability at most $e^{-cL^{2\beta}}$. It remains to evaluate $\mathbb{P}(\Omega^{\gamma,\gamma'}_{0,\beta}(\epsilon))$. We recall the following results from \cite{K11}. They were proved for the model defined in \eqref{fullopdis} with $a_n=1$, they extend readily to our case. First, we study the variations of the Jacobian. 
\begin{lem}\label{Hessien}
There exists $C>0$ such that
\begin{displaymath}
\Vert Hess_\omega(E(\omega))\Vert_{l^\infty\rightarrow l^1}\leq \dfrac{C}{dist\big[E(\omega),\sigma(H_\omega(\Lambda_l))-\{E(\omega)\}\big ]}.
\end{displaymath}
\end{lem}
Fix $\alpha\in(1/2,1)$. Using Lemma~\ref{Hessien} and (M) when $H_\omega(\Lambda_l)$ has two eigenvalue in $[E-L^{-\alpha},E+L^{-\alpha}]$, for L large enough, with probability at least $1-L^{-2\alpha}\lambda $,
\begin{equation}
\Vert Hess_\omega(E(\omega))\Vert_{l^\infty\rightarrow l^1}+\Vert Hess_\omega(E'(\omega))\Vert_{l^\infty\rightarrow l^1}\leq CL^\alpha
\end{equation}
 
In the following lemma we write $\omega=(\omega_\gamma,\omega_{\gamma'},\omega_{\gamma,\gamma'})$.
\begin{lem}\label{square}
Pick $\epsilon= L^{-\alpha}$. For any $\omega_{\gamma,\gamma'}$, if there exists $(\omega_\gamma^0,\omega_{\gamma'}^0)\in \R^2$ such that $(\omega_\gamma^0,\omega_{\gamma'}^0,\omega_{\gamma,\gamma'})\in\Omega^{\gamma,\gamma'}_{0,\beta}(\epsilon)$, then for $(\omega_\gamma,\omega_{\gamma'}) \in \R^2$ such that $|(\omega_\gamma,\omega_{\gamma'})-(\omega_\gamma^0,\omega_{\gamma'}^0)|_\infty\leq \epsilon$ one has 
\begin{displaymath}
(E_j(\omega),E_k(\omega))\in J_L\times J_L'\Longrightarrow |(\omega_\gamma,\omega_{\gamma'})-(\omega_\gamma^0,\omega_{\gamma'}^0)|_\infty\leq L^{-1}\lambda^{-2}.
\end{displaymath} 
\end{lem}

As in Lemma~\ref{square}, fix $(\omega_\gamma^0,\omega_{\gamma'}^0)$ such that $(\omega_\gamma^0,\omega_{\gamma'}^0,\omega_{\gamma,\gamma'})\in\Omega^{\gamma,\gamma'}_{0,\beta}(\epsilon)$ and define $\mathcal{A}:=(\omega_\gamma^0,\omega_{\gamma'}^0)+\{(\omega_\gamma,\omega_{\gamma'}) \in \R_+^2\cup \R_-^2 ,\left|\omega_\gamma\right|\geq \epsilon \text{ or } \left|\omega_{\gamma'}\right|\geq \epsilon \}$. We know that for any $i\in\Z$, $\omega_i\rightarrow E_j(\omega)$ and  $\omega_i\rightarrow E_k(\omega)$ are non increasing functions. Thus, if $ 
(\omega_\gamma,\omega_{\gamma'})\in\mathcal{A}$ then $(E_j(\omega),E_k(\omega))\notin J_L\times J_L'$. Thus, all the squares of side $\epsilon$ in which there is a point in $\Omega^{\gamma,\gamma'}_{0,\beta}(\epsilon)$ are placed along a non-increasing broken line that goes from the upper left corner to the bottom right corner. As the random variables are bounded by $C>0$, there is at most $C L^\alpha$ cubes of this type.

As the $(\omega_n)_n$ are i.i.d, using Lemma~\ref{square} in all these cubes, we obtain  : 
\begin{equation}\label{proba2}
\mathbb{P}(\Omega^{\gamma,\gamma'}_{0,\beta}(\epsilon))\leq CL^{\alpha-2}\lambda^{-4}
\end{equation}
and therefore
\begin{equation}
\mathbb{P}_\epsilon\leq CL^{\alpha-2}\lambda^{-3}.
\end{equation}
Optimization yields $\alpha=2/3$. This completes the proof of Theorem~\ref{thdec2}.

\subsection{Proof of Lemma~\ref{probcoli} for alloy-type models}\label{subseccont}
The proof of Lemma~\ref{probcoli} for the discrete alloy-type models is the same as the proof for alloy-type models on the continuum, using results of Appendix~\ref{FDEapp} and making modifications due to the discrete structure. Therefore, we only write the proof for the models defined in \eqref{fullopcont}. We divide the proof into two parts but we first introduce some definitions. Recall that $q$ is the simple-site potential and that it satisfies \eqref{stepfunction}. On $L^2(-N,N)$ we define the non-negative symmetric bi-linear form : 
\begin{equation}\label{sem-inn}
\langle f,g\rangle_q = \int_{-N}^N f(t)g(t) q(t) dt.
\end{equation}
We denote $\|.\|_q$ the corresponding semi-norm. We say that the functions $f$ and $g$ are $q$-orthogonal if $\langle f,g\rangle_q=0$. The notion of 1-orthogonality is the usual orthogonality in $L^2(-l,l)$. Fix $(F,G)\in\R$ and let $u$ and $v$ be 1-normalized eigenfunctions of $H_\omega(\Lambda_l)$ associated to the eigenvalues $E_j(\omega)\in[F-e^{-l^\beta},F+e^{-l^\beta}]$ and $E_k(\omega)\in[G-e^{-l^\beta},G+e{-l^\beta}]$. These eigenvalues are almost surely simple and we compute  
\begin{equation}\label{dercont}
\partial_{\omega_n} E_j(\omega) =\left\langle \left(\partial_{\omega_n} H_\omega\right) u,u \right\rangle_1= \|u_{|_{(n-N,n+N)}}\|^2_q>0.
\end{equation}
As $q$ is bounded and satisfy the covering condition (H2), there exists $C>1$ such that  for all $L>0$, $\|\nabla_\omega\big(E_j(\omega))\|\in[1/C,C]$. In the rest of the subsection, $M$ will be fixed such that $\mathbb{P}(|\omega_0|>M)=0$ so that all the random variables $(\omega_i)_i$ are almost surely bounded by $M$.

For $\bullet\in\{F,G\}$, define the following ODE
\begin{equation}
(\mathcal{E}^n_\bullet)\,: \,
\forall x\in(-N,N),y''(x)+V_\omega(n+x) y(x) = \bullet y(x)
\end{equation}
Let $(e_{1,\bullet}^n,e_{2,\bullet}^n)$ be a q-orthonormal basis of solutions of $(\mathcal{E}^n_\bullet)$. We will use the following proposition: 
\begin{prop}\label{anabasis}
Fix $\bullet\in\{F,G\}$. We can choose $e_{1,\bullet}^n$ and $e_{2,\bullet}^n$ so that they are analytic functions of the $(\omega_j)_{j\in\llbracket n-2N,n+2N\rrbracket}\in[-M,M]^{4N+1}$.
\end{prop}
\begin{proof}
We omit the dependence on $n$ and $\bullet$ and only write $\omega$ instead of \\$(\omega_j)_{j\in\llbracket n-N,n+N\rrbracket}$. Let $\Psi$ and $\Phi$ be the solutions of $(\mathcal{E}^n_j)$ satisfying $\Psi'(0)=\Phi(0)=0$ and $\Psi(0)=\Phi'(0)=1$. We know that $\Psi$ and $\Phi$ are power series of $\omega$ and that $\|\Psi\|_q\|\Phi\|_q \neq 0$. Thus, $e_1:=\dfrac{\Psi}{\|\Psi\|_q}$ is analytic and satisfies $(\mathcal{E}^n_j)$. Now, define $\tilde{\Phi}:=\Phi-\langle \Phi,e_1\rangle e_1$. Then, $\tilde{\Phi}$ is an analytic non-zero function orthogonal to $e_1$ satisfying   $(\mathcal{E}^n_j)$. This concludes the proof of Proposition~\ref{anabasis}, taking $e_1$ and $e_2:=\dfrac{\tilde{\Phi}}{\|\tilde{\Phi}\|_q}$.
\end{proof}
 Now, as $u$ satisfies the ODE
\begin{equation}
\forall x\in(-N,N),\,y''+V_\omega(n+x)y(x)=Fy(x)+(E_j(\omega)-F)y(x)
\end{equation}
 with $|E_j(\omega)-F|\leq e^{-l^\beta}$ ($v$ satisfies a similar ODE) there exist two unique couples $(A_n,B_n)\in\R^2$ and $(\tilde{A}_n,\tilde{B}_n)\in\R^2$ such that, for all $x\in(n-N,n+N)$, 
\begin{equation}
\left\{\begin{aligned}
u(x):=A_n e_{1,F}^n(x-n)+B_n e_{2,F}^n(x-n)+\epsilon_u^n(x-n)\\
v(x)=\tilde{A}_n e_{1,G}^n(x-n)+\tilde{B}_n e_{2,G}^n(x-n)+\epsilon_v^n(x-n)
\end{aligned}\right. .
\end{equation} 
 and such that for $\bullet\in\{u,v\}$ we have $\epsilon_\bullet^n(0)=\left(\epsilon_\bullet^n\right)'(0)=0$. We then have 
 \begin{displaymath}
 \|\epsilon_u^n\|_\infty+\|\epsilon_v^n\|_\infty+\|\left(\epsilon_u^n\right)'\|_\infty+\|\left(\epsilon_v^n\right)'\|_\infty\leq Ce^{-l^\beta}
 \end{displaymath}
 for some $C>0$ (depending only on $\|q\|_\infty$, $M$ and $N$). Therefore, \\
 $\|u_{|_{(n-N,n+N)}}\|_q^2=A_n^2+B_n^2+\varepsilon_n^u$ and $\|v_{|_{(n-N,n+N)}}\|_q^2=\tilde{A}_n^2+\tilde{B}_n^2+\varepsilon_n^v$ with $|\varepsilon_n^u|+|\varepsilon_n^v|\leq Ce^{-l^\beta}$. Thus, 
\begin{equation}
\left\{\begin{aligned}
\mathcal{N}:=\|\nabla E_j\|_1=\sum_{n=-l}^l (A_n^2+B_n^2)+\xi_u\\
\tilde{\mathcal{N}}:=\|\nabla E_k\|_1=\sum_{n=-l}^l (\tilde{A}_n^2+\tilde{B}_n^2)+\xi_v
\end{aligned}\right.
\end{equation}
with $|\xi_u|+|\xi_v|\leq Ce^{-l^\beta}$.
Now, define : $\left\{
\begin{aligned}
C_n:=\dfrac{A_n}{\sqrt{\mathcal{N}}}~~,~~\tilde{C}_n:=\dfrac{\tilde{A}_n}{\sqrt{\tilde{\mathcal{N}}}}\\
D_n:=\dfrac{B_n}{\sqrt{\mathcal{N}}}~~,~~\tilde{D}_n:=\dfrac{\tilde{B}_n}{\sqrt{\tilde{\mathcal{N}}}}
\end{aligned}
\right. $. Then, we have
\begin{equation}
\sum_{n=-l}^l C_n^2+D_n^2=\sum_{n=-l}^l \tilde{C}_n^2+\tilde{D}_n^2+O(e^{-l^\beta})=1+O(e^{-l^\beta}).
\end{equation}

Finally, define $U(n)=\begin{pmatrix}
C_n\\
D_n
\end{pmatrix}$ and $V(n)=\begin{pmatrix}
\tilde{C}_n\\
\tilde{D}_n
\end{pmatrix}$, define the Prüfer variables $(r_u,\theta_u)\in\R_+^*\times[0,2\pi)$ such that
$U(n)=r_u\begin{pmatrix}
\sin \theta_u\\
\cos \theta_u
\end{pmatrix}$ 	and define \\
$t_u:=\text{sgn}(\tan \theta_u) \inf\left(|\tan \theta_u|,|\cot \theta_u|\right)$ and the same for $t_v$. The function $t_u$ is equal to $\tan \theta_u$ or $\cot \theta_u$ depending on whether $|\tan \theta_u|\leq 1$ or $|\tan \theta_u|\geq 1$. Using these notations, \eqref{gradcoli} can be rewritten
\begin{equation}
\|r_u-r_v\|_1\leq Ce^{-l^\beta/2}.
\end{equation}

We first prove the 
\begin{lem}\label{closetan}
There exist nine analytic functions $(f_i)_{i\in\llbracket 0,8\rrbracket}$ (only depending on $q$ and $N$) defined on $\R^{16N+1}$ and not all constantly equal to zero such that, if $u$ (respectively $v$) is a 1-normalized eigenfunction of $H_\omega$ associated to $E_j(\omega)\in\left[F-e^{-l^\beta},F+e^{-l^\beta}\right]$ (respectively associated to $E_k(\omega)\in[G-e^{-l^\beta},G+e^{-l^\beta}]$), if for some $n_0\in\Z\cap(-l+8N,l-8N)$ we have $r_u(n_0)\geq e^{-l^\beta/4}$ and
\begin{displaymath}
\forall m\in\llbracket n_0-7N,n_0+7N\rrbracket, \left|r_u(m)-r_v(m)\right|\leq e^{-l^\beta/2} 
\end{displaymath}
and if we define the polynomials 
\begin{displaymath}
\mathcal{R}_{\hat{\omega}}(X):=\sum_{i=0}^8 f_i(\hat{\omega}) X^i\text{ and }
\mathcal{Q}_{\hat{\omega}}(X):=\sum_{i=0}^8 f_{8-i}(\hat{\omega}) X^i
\end{displaymath}
where we have defined $\hat{\omega}:=\left(\omega_{n_0-8N},\dots,\omega_{n_0+8N}\right)$, then we have : 
\begin{align*}
\text{if }\exists\, g\in\{\tan,\cot\}, 
\left\{\begin{aligned}
t_v(n_0)=g(\theta_v(n_0))\\
t_u(n_0)=g(\theta_u(n_0))
\end{aligned}\right., 
\text{ then }&\left|\mathcal{R}_{\hat{\omega}}\left(t_v\left(n_0\right)\right)\right|\leq e^{-l^\beta/4},\\
\text{otherwise, we have }  &\left|\mathcal{Q}_{\hat{\omega}}\left(t_v\left(n_0\right)\right)\right|\leq e^{-l^\beta/4}.
\end{align*}
\end{lem}
\begin{proof}
We will prove the result under the assumption $t_u(n_0)=\tan\theta_u(n_0)$ and $t_v(n_0)=\tan \theta_v(n_0)$, i.e when
\begin{equation}\label{assumptan}
\max\left(|\tan \theta_u(n_0)|,|\tan \theta_v(n_0)|\right)\leq 1.
\end{equation}  
There are minor modifications in the other cases. As the random variables are i.i.d, it suffices to show the result with $n_0=0$, which will be supposed from now on. We then consider the ODE
\begin{equation}
\forall x\in(-7N,7N),y''(x)+V_\omega(x) y(x) = F y(x),
\end{equation}
which depends only on $(\omega_{-8N},\dots,\omega_{8N})$. Suppose $|r_u(m)-r_v(m)|\leq e^{-l^\beta/2}$ for $m\in\llbracket -7N,7N\rrbracket$ and $r_u(0)\geq e^{-l^\beta/4}$. We show that $t_v(0)$ is almost a root of a polynomial depending only on $(\omega_{-8N},\dots,\omega_{8N})$. 

In the following lines $\varepsilon$ will denote a vector such that $\|\varepsilon\|\leq Ce^{-l^\beta/2}$, its value may change from a line to another. As $u$ and $v$ have continuous derivatives, 
\begin{equation}
M^F U(2N)=N^F U(0)+\varepsilon
\end{equation}
where $M^F:=\begin{pmatrix}
e_{1,F}^0(-N) & e_{2,F}^0(-N)\\
(e_{1,F}^0)'(-N) & (e_{2,F}^0)'(-N)
\end{pmatrix}$ and 
$N^F:=\begin{pmatrix}
e_{1,F}^0(N) & e_{2,F}^0(N)\\
(e_{1,F}^0)'(N) & (e_{2,F}^0)'(N)
\end{pmatrix}$. Thus, if we define $T_F^+:= (M^F)^{-1}N^F$ we have
\begin{equation}\label{defT}
U(2N)=T_F^+\,U(0) +\varepsilon\text{  and  } V(2N)=T_F^+\,V(0)+\varepsilon.
\end{equation}
Indeed, the matrix $(M_F)^{-1}$ depends only on $(\omega_{-N},\dots,\omega_N)\in[-M,M]^{2N+1}$ and is therefore uniformly bounded by a constant $C>0$ (depending only on $\|q\|_\infty$, $M$ and $N$).

As $t_u(0)=\tan\theta_u(0)$, we compute  
\begin{align*}
\left(\dfrac{r_u(2N)}{r_u(0)}\right)^2=&\left\|T_F^+\begin{pmatrix}
\sin \theta_u(n)\\
\cos \theta_u(n)
\end{pmatrix}\right\|^2+\epsilon \\
=&\dfrac{1}{1+t_u(n)^2}\left\|T_F^+\begin{pmatrix}
t_u(n) \\
1
\end{pmatrix}\right\|^2+\epsilon,
\end{align*}
for some $|\epsilon|\leq Ce^{-l^\beta/4}$.
In the case $t_u(0)=\dfrac{1}{\tan\theta_u(0)}$, we compute
\begin{align*}
\left(\dfrac{r_u(2N)}{r_u(0)}\right)^2=&\left\|T_F^+\begin{pmatrix}
\sin \theta_u(n)\\
\cos \theta_u(n)
\end{pmatrix}\right\|^2+\epsilon \\
=&\dfrac{1}{1+t_u(n)^2}\left\|T_F^+\begin{pmatrix}
1 \\
t_u(n)
\end{pmatrix}\right\|^2+\epsilon.
\end{align*}

The eigenvector $v$ satisfies the same equation if we replace $F$ by $G$. Therefore, the equation
\begin{displaymath}
\left|\left(\dfrac{r_u(2N)}{r_u(0)}\right)^2-\left(\dfrac{r_v(2N)}{r_v(0)}\right)^2 \right| \leq e^{-l^\beta/4}
\end{displaymath}
 can be rewritten
 \begin{equation}
\left|\dfrac{1}{1+t_u(0)^2}\left\|T_F^+\begin{pmatrix}
t_u(n) \\
1
\end{pmatrix}\right\|^2  \\
- \dfrac{1}{1+t_v(0)^2}\left\|T_G^+\begin{pmatrix}
t_v(n) \\
1
\end{pmatrix}\right\|^2\right|\leq Ce^{-l^\beta/4}.
\end{equation} 
Thus, there exists $\epsilon_1$ such that $|\epsilon_1|\leq Ce^{\-l^\beta/4}$ and such that
\begin{equation}\label{eqabove}
\dfrac{1}{1+t_u(0)^2}\left\|T_F^+\begin{pmatrix}
t_u(n) \\
1
\end{pmatrix}\right\|^2 \\
=\dfrac{1}{1+t_v(0)^2}\left\|T_G^+\begin{pmatrix}
t_v(n) \\
1
\end{pmatrix}\right\|^2+\epsilon_1.
\end{equation}
Now, consider the equation $U(-2N)=T_F^-U(n)+\varepsilon$ for the matrix $T_F^-$ constructed in the same way as $T_F^+$. Using the same calculations as to prove \eqref{eqabove} we obtain the existence of $\eta_1$ with $|\eta_1|\leq Ce^{\-l^\beta/4}$ such that 
\begin{equation}\label{eqbelow}
\dfrac{1}{1+t_u(0)^2}\left\|T_F^-\begin{pmatrix}
t_u(0) \\
1
\end{pmatrix}\right\|^2 \\
=\dfrac{1}{1+t_v(0)^2}\left\|T_G^-\begin{pmatrix}
t_v(0) \\
1
\end{pmatrix}\right\|^2+\eta_1.
\end{equation}
Define the polynomials of degree $2$
\begin{displaymath}
\begin{aligned}
P_G(t)&:=\left\|T_G^+\begin{pmatrix}
t \\
1
\end{pmatrix}\right\|^2
\end{aligned}
\text{ and  }
\begin{aligned}
Q_G^n&:=\left\|T_G^-\begin{pmatrix}
t \\
1
\end{pmatrix}\right\|^2.
\end{aligned} 
\end{displaymath}
Using \eqref{assumptan}, the equations \eqref{eqabove} and \eqref{eqbelow} can be rewritten
\begin{multline}\label{eqplus}
R_1(t_u(0),t_v(0)):=t_u^2(0)\left[(1+t_v^2(0))\left\|T_F^+\begin{pmatrix}
1 \\
0
\end{pmatrix}\right\|^2-P_G(t_v(0))\right] \\
+2t_u(0)\left\langle T_F^+\begin{pmatrix}
1 \\
0
\end{pmatrix},T_F^+\begin{pmatrix}
0 \\
1
\end{pmatrix} \right\rangle(1+t_v^2(0)) \\
+\left[(1+t_v^2(0))\left\|T_F^+\begin{pmatrix}
0 \\
1
\end{pmatrix}\right\|^2-P_G(t_v(0))\right]=\epsilon_2
\end{multline}
and 
\begin{multline}\label{eqminus}
R_2(t_u(0),t_v(0)):=t_u^2(n)\left[(1+t_v^2(0))\left\|T_F^-\begin{pmatrix}
1 \\
0
\end{pmatrix}\right\|^2-P_G(t_v(0))\right] \\
+2t_u(0)\left\langle T_F^-\begin{pmatrix}
1 \\
0
\end{pmatrix},T_F^-\begin{pmatrix}
0 \\
1
\end{pmatrix} \right\rangle(1+t_v^2(0)) \\
+\left[(1+t_v^2(0))\left\|T_F^-\begin{pmatrix}
0 \\
1
\end{pmatrix}\right\|^2-P_G(t_v(0))\right]=\eta_2.
\end{multline}

Thus, $t_u(0)$ is a root of the two polynomials $t\to R_1(t,t_v(0))-\epsilon_2$ and $t\to R_2-\eta_2$. Therefore, the resultant of these polynomials must be zero. 
All the coefficients in $R_1$ and $R_2$ are bounded uniformly over $(\omega_m)_{m\in\llbracket -N,N\rrbracket}$. Thus the resultant $\mathcal{R}(t_v(0))$ of $R_1(\,\cdot\,,t_v(0))$ and $R_2(\,\cdot\,,t_v(0))$ is smaller than $e^{-l^\beta/4}$. 

If we have $t_u(0)=\tan\theta_u(0)$ but $t_v(0)=\cot\theta_v(0)$ instead of $t_v(0)=\tan\theta_v(0)$, the resultant $ \mathcal{Q}(t)$ obtained is equal to $t^8\mathcal{R}(1/t)$.  If we have $t_u(0)=\cot\theta_u(0)$ and $t_u(0)=\tan\theta_u(0)$ instead of $t_u(0)=\tan\theta_u(0)$ and $t_v(0)=\tan\theta_v(0)$, the resultant obtained is $\mathcal{R}$

Now, we can study the pair of fractions $\dfrac{r_u(2iN)}{r_u(0)}$ and $\dfrac{r_u(-2iN)}{r_u(0)}$ for $i\in\N$ (the construction above is then the case $n=1$) and construct the resultant $\mathcal{R}_i$ in the same way we constructed $\mathcal{R}$ but where the operators $T_j^\pm$ are replaced by $(T_{j,i}^\pm)$, using the continuity of the derivatives at points $\{(2i+1)N, i\in\Z\}$. Now, the resultants $(\mathcal{R}_i)_{i\in\{1,2,3\}}$ are analytic functions of the random variables $\left(\omega_{-8N},\omega_{-8N+1},\dots,\omega_{8N}\right)$. We will now prove that one of these resultants is not constantly the zero polynomial. This will be done under the assumption $\omega_{-8N}=\omega_{-8N+1}=\dots=\omega_{8N}$. Under this assumption we have 
\begin{displaymath}
\forall x\in(-7N,7N),~ V_\omega(x)=\sum_{n\in\Z} \omega_n q(x-n)=\omega_0 \sum_{n\in(x-N,X+N)\cap\Z}q(x-n).
\end{displaymath}
 Therefore, we come down to the study of the ODEs
\begin{equation}\label{ODE}
\forall x\in(-7N,7N),y''(x)=\left(\omega_0 \tilde{q}(x)-\bullet\right) y(x)
\end{equation}
where $\tilde{q}$ is one-periodic, $\tilde{q}>0$ on $(-1/2,1/2)$ and $\bullet\in\{F,G\}$. These equations show that in this case, for $i\in\{1,2,3\}$ and $\bullet\in\{F,G\}$, we have the relations (see \eqref{defT})
\begin{displaymath}
T_{\bullet,i}^\pm=(T_\bullet^\pm)^i\text{ and } T_\bullet^-=(T_\bullet^+)^{-1}.
\end{displaymath}

We now prove the
\begin{lem}\label{resnotzero}
There exists $i_0\in\{1,2,3\}$ such that for $\omega_{0}$ such that $\omega_{0} \eta > F>G$ ($\eta$ is defined in (H2)), $\mathcal{R}_{i_0}$ is not the zero polynomial.
\end{lem}

\begin{proof}
Fix $\omega_{0}$ such that $\omega_{0} \eta > F>G$. The fact that the resultant $\mathcal{R}_{i_0}$ is the zero polynomial is equivalent to the fact that, for all $t'\in\R$, the polynomials $R_1(\,\cdot\, ,t')$ and $R_2(\,\cdot\, ,t')$ have a common root. This is also equivalent to the fact that for all $w\neq0$ satisfying $\eqref{ODE}$ for $\bullet=G$, there exists a function $z:=z(w)\neq 0$ satisfying \eqref{ODE} for $\bullet=F$ such that 
\begin{displaymath}
\left|\left(\dfrac{r_z(2N)}{r_z(0)}\right)^2-\left(\dfrac{r_w(2N)}{r_w(0)}\right)^2 \right|+ \left|\left(\dfrac{r_z(-2N)}{r_z(0)}\right)^2-\left(\dfrac{r_w(-2N)}{r_w(0)}\right)^2 \right|=0.
\end{displaymath}
Now, we remark that $r_z$ and $r_w$ do not change if we change the $q$-orthonormal bases $(e_{1,x}^n,e_{2,x}^n)$ for $x\in\{F,G\}$ by other $q$-orthonormal bases of solutions of $(\mathcal{E}_x^n)$. Therefore, if for $x\in\{F,G\}$, $(f_{1,x}^n,f_{2,x}^n)$  are other $q$-orthonormal bases of solution of $(\mathcal{E}_x^n)$ and if $(\mathcal{P}_i)_i$ are the resultants constructed in the same way that the $(\mathcal{R}_i)_i$ but in the bases $(f_{1,x}^n,f_{2,x}^n)$, the fact that the resultant $\mathcal{R}_i$ is the zero polynomial is equivalent to the fact that $\mathcal{P}_i$ is the zero polynomial.

 Take $q$-orthonormal bases $(f_{1,F},f_{2,F})$ and $(f_{1,G},f_{2,G})$ such that for $x\in\{F,G\}$, if $\tilde{q}_s(t):=\tilde{q}(t)+\tilde{q}(1-t)$
\begin{equation}\label{sym}
\int_{-N}^N f_{1,x}'(t) f_{2,x}'(-t)-f_{1,x}(t) f_{2,x}(-t) (\omega_0\tilde{q}_s(t)-E_x)dt = 0.
\end{equation}
We can now construct the Prüfer variables $(r_u,\hat\theta_u)$, $(r_v,\hat\theta_v)$, the operators $\hat{T}_j^\pm$ and $\hat{T}_k^\pm$ and the resultants $(\hat{\mathcal{R}}_i)_i$ in the bases $(f_{1,F},f_{2,F})$ and $(f_{1,G},f_{2,G})$. Using \eqref{eqplus} and \eqref{eqabove} (with $\hat{T}_x^\pm$ instead of $T_x^\pm$), the resultants $\hat{\mathcal{R}}_1$ is a polynomial of degree at most equal to 8 which leading coefficient is equal to the determinant of the matrix 
 \begin{equation}
A_1:=\begin{pmatrix}
\Delta_1 & 0 & \Delta_3 & 0\\
\Pi_+ & \Delta_1 & \Pi_- & \Delta_3\\
\Delta_2 & \Pi_+ & \Delta_3 & \Pi_-\\
0 & \Delta_2 & 0 & \Delta_4
\end{pmatrix}
\end{equation}
where we have defined 
\begin{displaymath}
\left\{\begin{aligned}
\Delta_1&:=\|\hat{T}_F^+ (1,0)\|^2 -\|\hat{T}_G^+ (1,0)\|^2,\\
\Pi_+&:=\left\langle \hat{T}_F^+ (1,0),\hat{T}_F^+ (0,1)\right\rangle,\\
\Delta_2&:=\|\hat{T}_F^+ (0,1)\|^2 -\|\hat{T}_G^+ (1,0)\|^2,\\
\Delta_3&:=\|\hat{T}_F^- (1,0)\|^2 -\|\hat{T}_G^- (1,0)\|^2\\
\Pi_-&:=\left\langle \hat{T}_F^- (1,0),\hat{T}_F^- (0,1)\right\rangle,\\
\Delta_4&:=\|\hat{T}_F^- (0,1)\|^2 -\|\hat{T}_G^- (1,0)\|^2.
\end{aligned}
\right. 
\end{displaymath}

In the same way, let $A_i$ be the matrix which coefficients $(\Delta_{m,i})_{m\in\llbracket 1,4 \rrbracket}$ and $\Pi_{\pm,i}$ are the same as the coefficients of $A_1$ but with $(\hat{T}_x^\pm)^i$ instead of $\hat{T}_x^\pm$. The leading coefficient of $\hat{\mathcal{R}}_i$ is the determinant of $A_i$.
We now show that, using \eqref{sym}, the coefficient of the matrices $(A_i)_i$ satisfy a relation of symmetry. 
\begin{lem}\label{symres}
Fix $x\in\{F,G\}$. For $i\in\N$, let $\begin{pmatrix}
a_x^i & b_x^i\\
c_x^i & d_x^i
\end{pmatrix}:=(\hat{T}_x^+)^i$. Then, $(\hat{T}_x^-)^i=(\hat{T}_x^+)^{-i}=\begin{pmatrix}
a_x^i & -b_x^i\\
-c_x^i & d_x^i
\end{pmatrix}$. Thus, $\Delta_{1,i}=\Delta_{3,i}$, $\Delta_{2,i}=\Delta_{4,i}$, $\Pi_{+,i}=-\Pi_{-,i}$.
\end{lem}
\begin{proof}
The Lemma is proven by induction if we prove it for $i=1$. We compute 
\begin{equation}
\left\{\begin{aligned}
a_x^1=f_{2,x}'(-N)f_{1,x}(N)-f_{2,x}(-N)f_{1,x}'(N)\\
b_x^1=f_{2,x}'(-N)f_{2,x}(N)-f_{2,x}(-N)f_{2,x}'(N)\\
c_x^1=f_{1,x}'(N)f_{1,x}(-N)-f_{1,x}(1)f_{1,x}'(-N)\\
d_x^1=f_{2,x}'(N)f_{1,x}(-N)-f_{2,x}(1)f_{1,x}'(-N)
\end{aligned}\right. ~~.
\end{equation}
The matrix $\hat{T}_x^-=(\hat{T}_x^+)^{-1}$ has the same coefficients where the arguments $N$ and $-N$ are exchanged. Hence, it remains to prove
\begin{equation*}
f_{2,x}'(-N)f_{1,x}(N)-f_{2,x}(N)f_{1,x}'(N)=f_{2,x}'(N)f_{1,x}(-N)-f_{2,x}(N)f_{1,x}'(-N).
\end{equation*}
Define  $W(t)=\begin{vmatrix}
f_{1,x}(t) & f_{2,x}(-t) \\
f_{1,x}'(t) & f_{2,x}'(-t)
\end{vmatrix}$. We compute, 
\begin{displaymath}
W'(t)= 2f_{1,x}'(t) f_{2,x}'(-t)-f_{1,x}(t) f_{2,x}(-t) (\tilde{q}(t)+\tilde{q}(-t)-2x).
\end{displaymath}  
Lemma~\ref{symres} is now a consequence of \eqref{sym}.
\end{proof}
We now  continue the proof of Lemma~\ref{resnotzero}. Using Lemma~\ref{symres}, for $i\in\llbracket 1,3\rrbracket$, we obtain
\begin{equation}
\det A_i = 4 \Pi_{+,i}^2 \Delta_{1,i} \Delta_{2,i}.
\end{equation}
We will use the Lemma~\ref{onezero} and the Lemma~\ref{threezero}. The solutions of \eqref{ODE} will be extended to $\R$  so that they satisfy \eqref{ODEapp}. For $(x,y)\in\{1,2\}\times\{F,G\}$, let $\mathcal{F}_{x,y}$ denote the extension of $f_{x,y}$ to $\R$ satisfying \eqref{ODEapp}. Then the components of $(\hat{T}_x^+)^i(1,0)$ are the coordinates of the restriction of $\mathcal{F}_{1,x}$ to $[(2i-1)N,(2i+1)N]$. We know from Lemma~\ref{onezero} that $\mathcal{F}_{1,x}$ and that $\mathcal{F}_{2,x}$ have at most one zero. As they are q-orthogonal on $(-N,N)$, we can suppose that $\mathcal{F}_{2,x}$ vanishes in $(-N,N)$ and that $\mathcal{F}_{2,x}$ is positive on $(N,\infty)$ and negative on $(-\infty,-N)$. 

Now, if $\Pi_+=0 ~(=-\Pi_-)$, then $F_{1,x}$ must have one zero in $(N,3N)$ and one in $(-3N,0)$. This is in contradiction with Lemma~\ref{onezero}. Thus, $\Pi_+\neq 0$. Now suppose that $\mathcal{F}_{1,x}$ vanishes at some point, for instance in $[(2m-1)N,(2m+1)N]$. Then, $\mathcal{F}_{1,x}$ is positive in $((2m+1)N,\infty)$ and negative in $(-\infty,(2m-1)N)$. Using Lemma~\ref{symres}, this contradict the fact that
\begin{displaymath}
\left\langle (\hat{T}_F^+)^m (1,0),(\hat{T}_F^+)^m (0,1)\right\rangle=-\left\langle (T_F^-)^m (1,0),(T_F^-)^m (0,1)\right\rangle.
\end{displaymath} 
Therefore, $\mathcal{F}_{1,x}$ has no zero and we can assume that $\mathcal{F}_{1,x}$ is positive.

Now, for $f:\R\to\R $, define $P_i(f)$ as the restriction of $f$ to $[(2i-1)N,(2i+1)N]$. Then, as $\|\hat{T}_{x,i}^\pm (1,0)\|=\|P_{\pm i}(\mathcal{F}_{1,x})\|_q$ (a similar equation hold for $(0,1)$ instead of $(1,0)$), for $i\in\{1,2,3\}$ we have
\begin{equation}
\left\{\begin{aligned}
\Delta_{1,i}=0\Longleftrightarrow \left\langle P_i(\mathcal{F}_{1,F}+\mathcal{F}_{1,G}),P_i (\mathcal{F}_{1,F}-\mathcal{F}_{1,G})\right\rangle_q=0\\
\Delta_{2,i}=0\Longleftrightarrow \left\langle P_i (\mathcal{F}_{2,F}+\mathcal{F}_{1,G}),P_i (\mathcal{F}_{2,F}-\mathcal{F}_{1,G})\right\rangle_q=0
\end{aligned}\right. ~~.
\end{equation}
For $(x,y)\in\{1,2\}^2$, we know from Lemma~\ref{threezero} that the function $\mathcal{F}_{y,F}-\mathcal{F}_{x,G}$ changes of sign at most three times. Now, $\mathcal{F}_{x,F}+\mathcal{F}_{y,G}$ is positive on $[N,+\infty)$ and negative on $(-\infty,-N]$. Suppose that for $i\in\{1,2,3\}$ either 
\begin{displaymath}
\left\langle P_i (\mathcal{F}_{1,F}+\mathcal{F}_{1,G}),P_i (\mathcal{F}_{1,F}-\mathcal{F}_{1,G})\right\rangle_q =0
\end{displaymath}
 or 
 \begin{displaymath}
\left\langle P_i (\mathcal{F}_{2,F}+\mathcal{F}_{1,G}),P_i (\mathcal{F}_{2,F}-\mathcal{F}_{1,G})\right\rangle_q =0.
\end{displaymath} 
Then, either $\mathcal{F}_{1,F}-\mathcal{F}_{1,G}$ or $\mathcal{F}_{2,F}-\mathcal{F}_{1,G}$ must vanish twice in $(N,7N)$. Now, for $x\in\{F,G\}$, $\|\hat{T}_{x,i}^+ X \|=\|\hat{T}_{x,i}^- X\|$. Therefore, it must also vanish twice in $(-7N,-N)$. This is in contradiction with Lemma~\ref{threezero}. This completes the proof of Lemma~\ref{resnotzero}.
\end{proof}
 We can now finish the proof of Lemma~\ref{closetan}. There exist $(\omega_{-8N},\dots,\omega_{8N})$ such that the coefficients of $\mathcal{R}_{i_0}$ are not all equal to zero. Now, write \\
 $\mathcal{R}_{i_0}(X)=\sum_{i=0}^8 f_i(\omega_{-8N},\dots,\omega_{8N})X^i$ where the $(f_i)_i$ are analytic. Then, one of the functions $(f_i)_i$  must be not constantly equal to zero. Besides, by construction, we have $|\mathcal{R}_{i_0}(t_v(0))|\leq e^{-l^\beta/4}$. This completes the proof of Lemma~\ref{closetan}.
\end{proof}
 
We now continue the proof of Lemma~\ref{probcoli}. Fix $u$ (respectively $v$) a 1-normalized eigenfunction of $H_\omega$ associated to $E_j(\omega)\in\left[F-e^{-l^\beta},F+e^{-l^\beta}\right]$ (respectively associated to $E_k(\omega)\in[G-e^{-l^\beta},G+e^{-l^\beta}]$) and suppose that $\|r_u-r_v\|_1\leq e^{-l^\beta/2}$. Let $\mathcal{R}_\omega(X):=\sum_{i=0}^8 f_i\left(\omega_{-8N},\cdots,\omega_{8N}\right) X^i$ be the polynomial given in Lemma~\ref{resnotzero} and $i_1\in\llbracket 0,8\rrbracket$ be the largest index such that $f_i$ is not constantly equal to zero. Then, using Theorem~\ref{Loja} and the fact that the random variables are bounded by $M$, we obtain the 
\begin{prop}\label{defA}
Let $\hat{x}:=(x_{-8N},\cdots,x_{8N})$ and
\begin{equation*}
\mathcal{A}:=\left\{\hat{x}\in[-M,M]^{16N+1}, |f_{i_1}(x_{-8N},\cdots,x_{8N})|\leq e^{-l^\beta/8}\right\}.
\end{equation*}
There exists $c\in(0,1)$ such that, for $l$ large enough, we have
\begin{equation*}
\mathbb{P}\left( (\omega_{-8N},\cdots,\omega_{8N})\in\mathcal{A}\right)\leq e^{-cl^\beta}.
\end{equation*}
\end{prop} 
On $\mathcal{A}^c$, let $(Z_i(\hat{x}))_{i\in\llbracket 1, i_1\rrbracket}\in\C^{i_1}$ be the roots of $\mathcal{R}_{\hat{x}}$, repeated according to their multiplicity,  and $Z_i=\infty$ for $i\in\llbracket i_1,8\rrbracket$. Then, the roots of $\mathcal{Q}_x$, defined in Lemma~\ref{resnotzero}, are the inverses of the $(Z_i)_i$ not equal to zero, with the convention $\infty^{-1}=0$. We now prove the 

\begin{prop}\label{roots}
Fix $\hat{x}\in\mathcal{A}^c$ and suppose $|\mathcal{R}_{\hat{x}}(t)|\leq e^{-l^\beta/4}$ for some \\
$t\in[-1,1]$. Then, there exists $i\in\llbracket 1,i_1\rrbracket$ such that
\begin{displaymath}
\left|t-Z_i(\hat{x})\right|\leq e^{-l^\beta/64}.
\end{displaymath}
\end{prop}
\begin{proof}
Fix $\hat{x}\in\mathcal{A}^c$ and write 
\begin{displaymath}
|\mathcal{R}(t)|=|f_{i_1}(\hat{x})|\prod_{i=1}^{i_1}\big|t-Z_i(\hat{x})\big|\leq e^{-l^\beta/4}.
\end{displaymath}
As $|f_i(\hat{x})|\geq e^{-l^\beta/8}$ and $i_1\leq 8$, one of the term in the product must satisfy
\begin{displaymath}
\left|t-Z_i(\hat{x})\right|\leq e^{-l^\beta/64}.
\end{displaymath}
This concludes the proof of Proposition~\ref{roots}.
\end{proof}
Note that, for all $t\in\R$ and $z\in\C$ we have $|t-\Re(z)|\leq |t-z|$. Now, suppose for instance that $r_v(0)$, $r_v(-8N-1)$ and $r_v(8N+1)$ are all greater than $e^{-l^\beta/4}$ and suppose $\Omega_-:=(\omega_{-16N-1},\cdots,\omega_{-1})\in\mathcal{A}^c$ and $\Omega_+:=(\omega_1,\cdots,\omega_{16N+1})\in\mathcal{A}^c$. Then, Lemma~\ref{closetan} shows that 
\begin{displaymath}
\left|\mathcal{\bullet}_{\Omega_-}(t_v(-8N-1))\right|+\left|\mathcal{\lozenge}_{\Omega_+}(t_v(8N+1))\right|\leq e^{-l^\beta/4} 
\end{displaymath}
 for some $(\bullet,\lozenge)\in\{\mathcal{R},\mathcal{Q}\}^2$, depending on whether $t_\square=\tan(\theta_\square)$ or $t_\square=\cot(\theta_\square)$ for $\square\in\{u,v\}$.  As $\Omega_-$ and $\Omega_+$ belong to $\mathcal{A}^c$, there exists $(i_-,i_+)\in\llbracket 1,8 \rrbracket^2$ such that
\begin{equation*}
\left|t_v(-8N-1)-\Re\left(Z_{i_-}^*(\Omega_-)\right)\right|+\left|t_v(8N+1)-\Re\left(Z_{i_+}^\sharp(\Omega_+)\right)\right|\leq e^{-l^\beta/64},
\end{equation*}
for some $(*,\sharp)\in\{1,-1\}^2$, depending on whether $t_\square=\tan(\theta_\square)$ or $t_\square=\cot(\theta_\square)$ for $\square\in\{u,v\}$. In the rest of the section, we will use the same notation $\theta$ for the class of $\theta$ in $\R/\pi\Z$. We endow $\mathbb{T}:=\R/\pi\Z$ with the usual distance, which will be noted $d$, obtained from the absolute value on $\R$. 

If we define $\Theta_{i,1}(\Omega_\pm):=\arctan\left(\Re\left[Z_i(\Omega_\pm)\right]\right)$, and \\
$\Theta_{i,-1}(\Omega_\pm):=\cot^{-1}\left(\Re\left[Z_i^{-1}(\Omega_\pm)\right]\right)$, we have
\begin{equation}\label{defTheta}
d\big(\theta_v(-8N-1),\Theta_{i_-,\star}(\Omega_-)\big)+d\big(\theta_v(8N+1),\Theta_{i_+,\sharp}(\Omega_+)\big)\leq Ce^{-l^\beta/64}.
\end{equation}
 Now, as $\theta_v(-8N-1)$ represents the direction of the orthgonal projection of $v$ in the space of solutions of $(\mathcal{E}_G^{-8N-1})$, we deduce from \eqref{defTheta} that the direction of the vector $(v(-8N-1),v'(-8N-1))$ is almost fixed. In this context, take the Prüfer variables $(R_v(.),\psi_v(.))\in\R_+\times\R)$ such that
\begin{displaymath}
\begin{pmatrix}
v(t)\\
v'(t)
\end{pmatrix}
=R_u(t)
\begin{pmatrix}
\sin(\psi_v(t))\\
\cos(\psi_v(t))
\end{pmatrix}
\end{displaymath}
and define the matrix $\mathcal{A}_n:=\begin{pmatrix}
e^n_{1,G}(0) & e^n_{2,G}(0)\\
(e^n_{1,G})'(0) & (e^n_{2,G})'(0)
\end{pmatrix}$. We will also use the bi-Lipschitz homeomorphism $\Upsilon:\theta\in\mathbb{T}\to S^1/\{1,-1\}$, defined by $\theta\to(\sin\theta,\cos\theta)$, where we endowed $S^1/\{1,-1\}$ with the distance obtained from the euclidean distance on $\R^2$. We now prove the
\begin{prop}\label{transpruf}
Suppose $r_v(n)\geq e^{-l^\beta/4}$ and suppose there exists $\Theta\in\mathbb{T}$ such that $d\big(\theta_v(n);\Theta\big)\leq e^{-l^\beta/64}$. Then, there exists $C>1$ (only depending on $\|q\|_\infty$, $N$ and $M$) such that $R_u(n)\geq Cr_v(n)$ and there exists $\Psi\in\mathbb{T}$ (only depending on $\Theta$) such that $d(\psi_v(n);\Psi)\leq Ce^{-l^\beta/32}$.
\end{prop} 
\begin{proof}
Without loss of generality, we can suppose that $n=0$. By definition, we have for $x\in(-N,N)$
\begin{displaymath}
v(x)=\tilde{A}_n e_{1,G}^n(x)+\tilde{B}_n e_{2,G}^n(x)+\epsilon_v^n(x)
\end{displaymath}
with $\epsilon_v^0(0)=(\epsilon_v^0)'(0)=0$. Therefore, we obtain

\begin{displaymath}
\begin{pmatrix}
v(0)\\
v'(0)
\end{pmatrix}
=r_v(0)\mathcal{A}_0
\begin{pmatrix}
\sin(\theta_v(0))\\
\cos(\theta_v(0))
\end{pmatrix}
\end{displaymath}
We know that $\mathcal{A}_0$ depends only on $(\omega_{-N},\cdots,\omega_{N})$ and that $\det(\mathcal{A}_0)\neq 0$. Hence, there exists $C>0$ only depending on $\|q\|_\infty$, $N$ and $M$, such that
\begin{displaymath}
\forall (\omega_{-N},\cdots,\omega_{N})\in[-M,M]^{2N+1},~ \det(\mathcal{A}_0)\geq \dfrac{1}{C}.
\end{displaymath}
Therefore, if $X:=\begin{pmatrix}
\sin(\theta_v(0))\\
\cos(\theta_v(0))
\end{pmatrix} $
 we have $ \dfrac{1}{C}\leq \|\mathcal{A}_0X\|\leq C$. As $r_v(0)\geq e^{-l^\beta/4}$, for $l$ large enough we have $R_v(0)\geq \dfrac{1}{C}r_v(0)$. 

Now, we know that
\begin{displaymath}
\begin{pmatrix}
\sin(\psi_v(0))\\
\cos(\psi_v(0))
\end{pmatrix}
=\dfrac{r_v(0)}{R_v(0)}\mathcal{A}_0
X
\end{displaymath}
Thus, if $Y:=\dfrac{r_v(0)}{R_v(0)}\mathcal{A}_0
X$, we have $\|Y\|=1$ and  
\begin{equation}\label{closeangle}
\psi_v(0)=\Upsilon^{-1}\left(Y\right)=\Upsilon^{-1}\left(\dfrac{A_0X}{\|A_0X\|}\right)
\end{equation}

By definition, $X\in S^1$ is a representative of $ \Upsilon(\theta_v(0))$ and one can take a representative $\tilde{X}\in S^1$ of $\Upsilon(\Theta)$ satisfying $\|X-\tilde{X}\|\leq e^{-l^\beta/64}$. Thus, $\|\mathcal{A}_0(X-\tilde{X})\|\leq Ce^{-l^\beta/32}$ and $ \left\| \dfrac{\mathcal{A}_0X}{\|\mathcal{A}_0X\|}-\dfrac{\mathcal{A}_0\tilde{X}}{\|\mathcal{A}_0\tilde{X}\|}\right\|\leq Ce^{-l^\beta/32}$. Hence, if $\Psi:=\Upsilon^{-1}\left(\dfrac{\mathcal{A}_0\tilde{X}}{\|A_0\tilde{X}\|}\right)$ we obtain
\begin{displaymath}
d\left(\psi_v(0),\Psi\right)\leq Ce^{-l^\beta/64}.
\end{displaymath}
This achieve the proof of Proposition~\ref{transpruf}.
 \end{proof} 
 
Using Proposition~\ref{transpruf} and \eqref{defTheta}, there exists $\Psi_{i_-,\star}(\Omega_-)$ and $\Psi_{i_+,\sharp}(\Omega_+)$ such that 
\begin{equation}\label{defPsi}
d\big(\psi_v(-8N-1),\Psi_{i_-,\star}(\Omega_-)\big)+d\big(\psi_v(8N+1),\Psi_{i_+,\sharp}(\Omega_+)\big)\leq Ce^{-l^\beta/64}.
\end{equation}

 We will use the
\begin{prop}\label{defitrans}
Let  $W$ be a bounded function on $(x_1,x_2)\subset\R$. Let $(h_1,h_2)$ be the solutions of the ODE: 
\begin{equation}\label{ODEtrans}
\forall x\in(x_1,x_2),\,-y''(x)+W(x)y(x)= 0\,,
\end{equation}
satisfying $h_1(x_1)=h_2'(x_1)=1$ and $h_1'(x_1)=h_2(x_1)=0$. For all $\theta\in[0,2\pi)$, define $y_\theta:=\sin(\theta) h_1+\cos(\theta) h_2$. Then, the operator
\begin{displaymath}
\begin{aligned}
\mathcal{T}_{x_1,x_2,W}~:~\mathbb{T}&\to\mathbb{T}\\
\Theta&\mapsto\Upsilon^{-1}\left[\dfrac{1}{\left[y_\Theta(x_2)\right]^2+\left[y_\Theta'(x_2)\right]^2}\begin{pmatrix}
y_\Theta(x_2)\\
y_\Theta'(x_2)
\end{pmatrix}
\right]
\end{aligned}
\end{displaymath}
 is well-defined and Lipchitz continuous with $Lipchitz$ constant only depending on $\|W\|_\infty$ and $|x_1-x_2|$. Furthermore, there exists $C>0$ (only depending on $\|W\|_\infty$ and $|x_1-x_2|$) such that, for all $0<\epsilon<1$ and $\tilde{W}$ such that $\|\tilde{W}-W\|_\infty\leq \epsilon$, we have
 \begin{displaymath}
 \sup_{\Theta\in\mathbb{T}} d\big(\mathcal{T}_{x_1,x_2,W}(\Theta),\mathcal{T}_{x_1,x_2,\tilde{W}}(\Theta)\big)\leq C\epsilon.
\end{displaymath}  
\end{prop}
\begin{rem}
We postpone the proof of Proposition~\ref{defitrans} to the end of the subsection and finish the proof of Lemma~\ref{probcoli}, but first we make a remark. The number $\mathcal{T}_{x_1,x_2,W}(\Theta)$ is the direction of the vector $(y(x_2),y'(x_2))$ where $y$ can be any non-zero solution of the ODE
\begin{displaymath}
\forall x\in(x_1,x_2),\,-y''(x)+W(x)y(x)= 0\, ,
\end{displaymath}
such that the direction of the vector $(y(x_1),y'(x_1))$ is equal to any representative of $\Theta$. 
\end{rem}
Now, as the potential $V_\omega$ depends only on $(\omega_{-9N-1},\cdots,\omega_{-1})$ (which are fixed since $\Omega_-$ is fixed) over the interval $(-8N-1,-N)$, if \\
$\Phi_{i_-,\star}(\Omega_-):=\mathcal{T}_{-8N-1,-N,V_\omega-G}\left(\Psi_{i_-,\star}(\Omega_-)\right)$ we have
\begin{equation}
d\big(\psi_v(-N),\Phi_{i_-,\star}(\Omega_-)\big)\leq Ce^{-l^\beta/64}.
\end{equation}
In the same way, there exists $\Phi_{i_+,\sharp}(\Omega_+)$ such that 
\begin{equation}
d\big(\psi_v(N),\Phi_{i_+,\sharp}(\Omega_+)\big)\leq Ce^{-l^\beta/64}.
\end{equation}
Now, by definition, we have $\psi_v(N)=\mathcal{T}_{-N,N,V_\omega-E_k(\omega)}\left(\psi_v(-N)\right)$. Hence, as $|G-E_k(\omega)|\leq e^{-l^\beta}$, we obtain from Proposition~\ref{defitrans} that
\begin{displaymath}
d\big(\psi_v(N),\mathcal{T}_{-N,N,V_\omega-G}\left(\psi_v(-N)\right)\big)\leq Ce^{-l^\beta/64}
\end{displaymath}
and eventually
\begin{equation}\label{relPhi}
d\big(\Phi_{i_+,\sharp}(\Omega_+),\mathcal{T}_{-N,N,V_\omega-G}\left(\Phi_{i_-,\star}(\Omega_-)\right)\big)\leq Ce^{-l^\beta/64}.
\end{equation}

Now, as $\Omega_-$ and $\Omega_+$ are fixed, we can rewrite on the random potential $V_\omega$ on $(-N,N)$ as 
\begin{displaymath}
\forall x\in(-N,N),~V_\omega(x):=W(x)+\omega_0 q(x)
\end{displaymath}
where $W$ is a deterministic function. Thus, \eqref{relPhi} roughly says that the image by the random function $\mathcal{T}_{\omega_0}:=\mathcal{T}_{-N,N,W+\omega_0 q}$ of the fixed direction $\Phi_{i_-,\star}(\Omega_-)$ is almost fixed. This is a condition on the $32N+3$ random variables $(\Omega_-,\omega_0,\Omega_+)$ and the following lemma shows that this condition happens with exponentially small probability.
\begin{lem}\label{transtan}
Let $(\Theta_1,\Theta_2)\in \mathbb{T}^2$ and $W$ be a deterministic, bounded function. Define $\hat{\mathbb{P}}_\epsilon$ as the probability that $d\big(\mathcal{T}_{\omega_0}(\Theta_1),\Theta_2\big)\leq \epsilon$. Then, there exist $\epsilon_0>0$ and $M\in\N$ such that for all $0<\epsilon<\epsilon_0$, we have 
\begin{displaymath}
\hat{\mathbb{P}}_\epsilon\leq \left(\dfrac{\epsilon}{\epsilon_0}\right)^{1/M}.
\end{displaymath}
\end{lem}
\begin{proof}
Using Theorem~\ref{Loja2}, it suffices to show that if $y_{\Theta_1}$ is a non zero solution as in Proposition~\ref{defitrans} (with $x_1=-N$ and $x_2=N$), then the non-zero vector $(y_{\Theta_1}(N),y'_{\Theta_1}(N))$ is analytic in $\omega_0$ and its direction is not constant. As $q\geq0$ and is bounded below by a positive real number on an interval of positive length, this is proven by Sturm-Liouville theory (see \cite[Section 4.5]{Z05}). Indeed, there exist Prüfer variables $(R(t),\phi(t))$ such that $\begin{pmatrix}
y_{\Theta_1}(t)\\
y_{\Theta_1}'(t)
\end{pmatrix}=
R(t)\begin{pmatrix}
\sin(\phi(t))\\
\cos(\phi(t))
\end{pmatrix}$, such that $\phi(-N)$ is a representative of $\Theta_1$ and such that $\phi(N)$ is a strictly increasing function of $\omega_0$. As $y_{\Theta_1}(N)$ and $y_{\Theta_1}'(N)$ are analytic in $\omega_0$, this completes the proof of Lemma~\ref{transtan}.
\end{proof}
We can now finish the proof of Lemma~\ref{probcoli}. First, as $u$ is normalized and $u$ is at most exponentially increasing, there exists $n_0\in\Z$ such that
\begin{equation}
\forall n\in\llbracket n_0-l^\beta,n_0+l^\beta\rrbracket, r_u(n)\geq e^{-l^\beta/4}. 
\end{equation}
Note that $n_0$ is random but there are only $l$ choices for $n_0$. Therefore
\begin{multline*}
\mathbb{P}\left(\forall n\in (-l,l)\cap\Z ,\left|r_u(n)-r_v(n)\right|\leq e^{-l^\beta}\right)\\
\leq l\cdot\sup_{n_0\in(-l,l)\cap\Z}\mathbb{P}\left(\forall n\in(n_0-l^\beta,n_0+l^\beta)\cap\Z,
\begin{aligned}
&\left|r_u(n)-r_v(n)\right|\leq e^{-l^\beta/2}\\
&\text{ and } r_u(n)\geq e^{-l^\beta/4}
\end{aligned}
\right).
\end{multline*}
Now, fix $n_0\in(-l,l)\cap\Z$. To simplify our notations, let us assume that $n_0=0$. Define $K:=\left\lfloor l^\beta/(32N+3)\right\rfloor$ and define
\begin{displaymath}
I:=\bigsqcup_{j=-K}^{K}\llbracket (2j-1)(16N+1),(2j+1)(16N+1)\rrbracket=:\bigsqcup_{j=-K}^{K} J_j\subset (-l^\beta,l^\beta).
\end{displaymath}
Thus, if $m_j:=2j(16N+1)$ we have $J_j=\llbracket m_j-16N-1,m_j+16N+1\rrbracket$.
Define the event $\mathcal{B}_j:=\left\{\forall n\in J_j,\begin{aligned}
&\left|r_u(n)-r_v(n)\right|\leq e^{-l^\beta/2}\\
&\text{ and } r_u(n)\geq e^{-l^\beta/4}
\end{aligned}\right\}$ and define
\begin{align*}
\Omega_j^-&:=\left(\omega_{m_j-16N-1},\cdots,\omega_{m_j-1}\right),\\
\Omega_j^+&:=\left(\omega_{m_j+1},\cdots,\omega_{m_j+16N+1}\right).
\end{align*}
Then, using Lemma~\ref{resnotzero}, 
\begin{displaymath}
\mathcal{B}_j\subset 
\left\{\exists (\bullet,\lozenge)\in\{\mathcal{R},\mathcal{Q}\}^2,
\begin{aligned}
\left|\bullet_{\Omega_j^-}\big[t_v(m_j-8N-1)\big]\right|\leq e^{-l^\beta/4}\\
\left|\lozenge_{\Omega_j^+}\big[t_v(m_j+8N+1)\big]\right|\leq e^{-l^\beta/4}
\end{aligned}
\right\}.
\end{displaymath}
Now, using \eqref{defA} and \eqref{defTheta}, we obtain that $\mathcal{B}_j$ is included in
\begin{displaymath}
\left\{
\begin{aligned}
&\hspace{3cm}\Omega_j^-\in\mathcal{A}\text{ or } \Big[\Omega_j^-\in\mathcal{A}^c\text{ and }\\
&\exists (i,*)\in\llbracket 1,8\rrbracket\times\{-1,1\}, d\big(\theta_v(m_j-8N-1),\Theta_{i,\star}(\Omega_j^-)\big)\leq Ce^{-l^\beta/64}~~\Big], \\
&\hspace{3cm}\Omega_j^+\in\mathcal{A}\text{ or } \Big[\Omega_j^+\in\mathcal{A}^c\text{ and }\\
&\exists (i',\sharp)\in\llbracket 1,8\rrbracket\times\{-1,1\}, d\big(\theta_v(m_j+8N+1),\Theta_{i,\sharp}(\Omega_j^+)\big)\leq e^{-l^\beta/64}~~\Big]
\end{aligned}
\right\}.
\end{displaymath}
Eventually, using the notations of \eqref{relPhi}, we have $\mathcal{B}_j\subset \mathcal{C}_j$ where
\begin{displaymath}
\mathcal{C}_j:= 
\left\{
\begin{aligned}
&\Omega_j^-\in\mathcal{A}\text{ or }\Omega_j^+\in\mathcal{A}\text{ or } \Big[ (\Omega_j^-,\Omega_j^+)\in(\mathcal{A}^c)^2\text{ and }\exists (i,i')\in\llbracket 1,8\rrbracket^2,\\ 
&\exists(*,\sharp)\in\{-1,1\}^2,\,
d\big(\mathcal{T}_{\omega_{m_j}}(\Phi_{i,\star}(\Omega_j^-)),\Phi_{i,\sharp}(\Omega_j^+)\big)\leq Ce^{-l^\beta/64}
\end{aligned}
\right\}~.
\end{displaymath}
To summary, we have proven that 
\begin{multline}
\mathbb{P}
\left(
\forall n\in \llbracket -16N+1,16N+1\rrbracket,\left|\dfrac{r_u(n)}{r_u(0)}-\dfrac{r_v(n)}{r_v(0)}\right|\leq e^{-l^\beta/4}\right)\\
\leq \mathbb{P}
\left(\bigcap_{j=-K}^K \mathcal{B}_j
\right)\leq \mathbb{P}
\left(\bigcap_{j=-K}^K \mathcal{C}_j
\right)
\leq \prod_{j=-K}^K \mathbb{P}\left(C_j\right)
\end{multline}
because the events $\left(\mathcal{C}_j\right)_j$ are independent. Now, using Proposition~\ref{defA} and \\Lemma~\ref{transtan}, we have
\begin{equation}
\mathbb{P}(\mathcal{C}_j)\leq Ce^{-\tilde{c} l^\beta}
\end{equation}
for some $\tilde{c}\in(0,1)$. Therefore, we have 
\begin{multline}\label{transtanbis}
\mathbb{P}\left(\forall n\in\llbracket -16N-2,16N+2 \rrbracket,\left|\dfrac{r_u(n)}{r_u(0)}-\dfrac{r_v(n)}{r_v(0)}\right|\leq e^{-l^\beta/4}\right) \\
\leq C\left(e^{-\tilde{c}l^\beta}\right)^{2K+1}\leq e^{-\tilde{c}l^{2\beta}}
\end{multline}
for $l$ large enough since $K\asymp l^\beta$. This concludes the proof of Lemma~\ref{probcoli} in the continuous case. We now prove Proposition~\ref{defitrans}.
\begin{proof}[Proof of Proposition~\ref{defitrans}]
Fix $\theta\in\R$ and take $y_\theta$ as in the statement. Then, 
\begin{displaymath}
Y_\theta:=
\begin{pmatrix}
y_\theta(x_2)\\
y_\theta'(x_2)
\end{pmatrix}
=\begin{pmatrix}
h_1(x_2) & h_2(x_2)\\
h_1'(x_2) & h_2'(x_2)
\end{pmatrix}
\begin{pmatrix}
\sin(\theta)\\
\cos(\theta)
\end{pmatrix}
= \mathcal{H}\begin{pmatrix}
\sin(\theta)\\
\cos(\theta)
\end{pmatrix}
\end{displaymath}
with $\det \mathcal{H}=1$. As $h_1$ and $h_2$ satisfy \eqref{ODEtrans}, there exists $C>0$ (depending only on $\|W\|$ and $|x_2-x_1|)$ such that $\max\{\|\mathcal{H}\|,\|\mathcal{H}^{-1}\|\}\leq C$. Thus, for all $\theta\in\R$ we have $\|Y_\theta\|\geq \dfrac{1}{C}$, hence $\Xi_\theta:=\dfrac{Y_\theta}{\|Y_\theta\|}\in S^1$. Now, $y_{\theta}=-y_{\theta+\pi}$, hence $\Xi_{\theta}=-\Xi_{\theta+\pi}$ and $\mathcal{T}_{x_1,x_2,W}$ is well defined. 

Now, since $\theta\in\R \mapsto Y_\theta\in\R^2 $ is Lipchitz continuous and $C\geq \|Y_\theta\|\geq \dfrac{1}{C}$ for all $\theta$, the function $\theta\in\R \mapsto \Xi_\theta\in S^1$ is Lipschitz continuous. Thus, the fact that $\mathcal{T}_{x_1,x_2,W} $ is Lipschitz continuous follows from the fact that $\Upsilon^{-1}$ is also Lipschitz continuous.   

Eventually, take $\tilde{W}$ such that $\|\tilde{W}-W\|_\infty\leq \epsilon<1$ and fix $\theta\in\R$. We will use the same notations as above for $\tilde{W}$, but with the symbol $\tilde{\cdot}$. For instance, $(h_1,h_2)$ are the fundamental solutions of \eqref{ODEtrans} for $W$, $(\tilde{h}_1,\tilde{h}_2)$ are the fundamental solutions for $\tilde{W}$ and we have 
\begin{displaymath}
\mathcal{H}=\begin{pmatrix}
h_1(x_2) & h_2(x_2)\\
h_1'(x_2) & h_2'(x_2)
\end{pmatrix}\text{  and  }\tilde{\mathcal{H}}=\begin{pmatrix}
\tilde{h}_1(x_2) & \tilde{h}_2(x_2)\\
\tilde{h}_1'(x_2) & \tilde{h}_2'(x_2)
\end{pmatrix}.
\end{displaymath}
Then, we know that $\|\mathcal{H}-\tilde{\mathcal{H}}\|\leq C\epsilon$ for some $C$ depending only on $\|W\|$ and $|x_2-x_1|$. Therefore, there exists $C>0$, such that for all $\theta\in\R$ we have
\begin{displaymath}
\|Y_\theta-\tilde{Y}_\theta\|\leq C\epsilon\text{ and } \dfrac{1}{C}\leq\min\{\|Y_\theta\|,\|\tilde{Y}_\theta\|\}\leq \max\{\|Y_\theta\|,\|\tilde{Y}_\theta\|\}\leq C
\end{displaymath}
Therefore, for all $\theta\in\R$ we have $\|\Xi_\theta-\tilde{\Xi}_\theta\|\leq C\epsilon$. We conclude using the fact that $\Upsilon^{-1}$ is Lipschitz continuous. This achieve the proof of Proposition~\ref{defitrans}.
\end{proof}
 
\subsection{Proof of Lemma~\ref{probcoli} for discrete models}\label{subsecdis}
In this section we prove Lemma~\ref{probcoli} for models defined in \eqref{discmod}. When $\Delta_b=\Delta$, we can write a proof that is the same as in Subsection~\ref{subseccont}, except for the obvious modification due to the discrete models and the usage of results in Appendix~\ref{FDEsec}. In this section, we prove Lemma~\ref{probcoli} for models defined in \eqref{discmod}. For discrete models the functions that were analytic in the continuum case are here polynomials. This makes the study  of the solutions when the parameters tend to infinity easier.

As in subsection~\ref{subseccont} we start by making some notation, there are discrete equivalents of the notation find in the beginning of subsection~\ref{subseccont}. Define the following inner product
on $\ell^2\left(\llbracket 0,N-1\rrbracket\right)$  : 
\begin{equation}\label{sem-inn2}
\langle f,g\rangle_a = \sum_{m=0}^{N-1} a_mf(m)g(m) 
\end{equation}
We denote $\|.\|_a$ the corresponding semi-norm. We say the functions $f$ and $g$ are $a$-orthogonal is $\langle f,g\rangle_a=0$. In particular, 1-orthogonality is the usual orthogonality in $\ell^2\left(\llbracket 0,l\rrbracket\right)$.

Let $u$ and $v$ be 1-normalized eigenfunctions of $H_\omega(\Lambda_l)$ associated to the eigenvalues $E_j(\omega)\in[F-e^{-l^\beta},F+e^{-l^\beta}]$ and $E_k(\omega)\in[G-e^{-l^\beta},G+e^{-l^\beta}]$. These eigenvalues are almost surely simple and we compute : 
\begin{equation}\label{dercont2}
\partial_{\omega_n} E_j(\omega_n)= \left\langle \left(\partial_{\omega_n} H_\omega\right) u,u \right\rangle_1=\sum_{m=0}^{N-1} a_n u^2(m+Nn).
\end{equation}
Equation \eqref{dercont2} can therefore be rewritten : 
\begin{equation}
\partial_n E_j(\omega) = \|u_{|_{\llbracket Nn,Nn+k-1\rrbracket}}\|^2_a.
\end{equation}
In the rest of the subsection, $M$ will be fixed such that $\mathbb{P}(|\omega_0|>M)=0$ so that all the random variables $(\omega_i)_i$ are almost surely bounded by $M$.

On $\R^N$ we define the vectorial plan $\mathcal{P}^n_F$ defined by $\forall m\in\llbracket 1,N-2\rrbracket,$ 
\begin{equation}\label{plan}
b_{m+1}y(m+1)+b_{m}y(m-1)+(a_m \omega_n-F)y(m) = 0.
\end{equation}
In particular, when $N=2$, $\mathcal{P}^n_F=\R^2$ and the conditions below are trivial.
When $\omega_n$ tends to infinity the plan $\mathcal{P}^n_j$  get closer to the plan defined by the $a$-orthonormal basis $\left(e_1^\infty:=\dfrac{1}{\sqrt{a_{N-1}}}\delta_{N-1},e_2^\infty:=\dfrac{1}{\sqrt{a_0}}\delta_0\right)$. Thus, we can choose an $a$-orthonormal basis $(e_{1,F}^n,e_{2,F}^n)$ of the plan $(\mathcal{P}^n_F)$ such that 
\begin{equation}
\|e_{1,F}^n-e_1^\infty\|_\infty+\|e_{2,F}^n-e_2^\infty\|_\infty=o(1).
\end{equation}
The sequence $u$ satisfies the following equations for $m\in\llbracket Nn,Nn+N-1\rrbracket$ 
\begin{equation}
b_{m+1}u(m+1)+b_m u(m-1)+(a_n\omega_n-F)y(m)=(E_j(\omega)-F)y(m)
\end{equation}
with $|E_j(\omega)-F|\leq e^{-l^\beta}$ and $v$ satisfies a similar equation with $(E_k,G)$ instead of $(E_j,F)$. Therefore, there exist two couples $(A_n,B_n)\in\R^2$ and $(\tilde{A}_n,\tilde{B}_n)\in\R^2$ such that, for all $m\in\llbracket Nn,Nn+N-1\rrbracket$, 
\begin{equation}
\left\{\begin{aligned}
u(m):=A_n e_{1,F}^n(m-Nn)+B_n e_{2,F}^n(m-Nn)+\epsilon_u^n(m-Nn)\\
v(m)=\tilde{A}_n e_{1,G}^n(m-Nn)+\tilde{B}_n e_{2,G}^n(m-Nn)+\epsilon_v^n(m-Nn)
\end{aligned}\right. ,
\end{equation} 
with $\epsilon_\bullet^n(0)=\epsilon_\bullet^n(1)=0$. Note that $\|\epsilon_u^n\|_\infty+\|\epsilon_v^n\|_\infty\leq Ce^{-l^\beta}$ for some $C$ only depending on $\|a\|_\infty$, $M$ and $N$. 
Therefore, $\|u_{|_{\llbracket Nn,Nn+N-1\rrbracket}}\|_a^2=A_n^2+B_n^2+\varepsilon_u^n$ and $\|v_{|_{\llbracket kn,kn+k-1\rrbracket}}\|_a^2=\tilde{A}_n^2+\tilde{B}_n^2+\varepsilon_v^n$ with $|\varepsilon_u^n|+|\varepsilon_v^n|\leq Ce^{-l^\beta}$. Thus, 
\begin{equation}
\left\{\begin{aligned}
\mathcal{N}:=\|\nabla E_j\|_1=\sum_{n=1}^l (A_n^2+B_n^2)+\xi_u\\
\tilde{\mathcal{N}}:=\|\nabla E_k\|_1=\sum_{n=1}^l (\tilde{A}_n^2+\tilde{B}_n^2)+\xi_v
\end{aligned}\right. .
\end{equation}

Now, define : $\left\{
\begin{aligned}
C_n:=\dfrac{A_n}{\sqrt{\mathcal{N}}}~~,~~\tilde{C}_n:=\dfrac{\tilde{A}_n}{\sqrt{\tilde{\mathcal{N}}}}\\
D_n:=\dfrac{B_n}{\sqrt{\mathcal{N}}}~~,~~\tilde{D}_n:=\dfrac{\tilde{B}_n}{\sqrt{\tilde{\mathcal{N}}}}
\end{aligned}
\right. $. Then, we have
\begin{equation}
\sum_{n=1}^l C_n^2+D_n^2=\sum_{n=1}^l \tilde{C}_n^2+\tilde{D}_n^2+O(e^{-l^\beta})=1+O(e^{-l^\beta}).
\end{equation}

Now, define $U(n)=\begin{pmatrix}
C_n\\
D_n
\end{pmatrix}$ and $V(n)=\begin{pmatrix}
\tilde{C}_n\\
\tilde{D}_n
\end{pmatrix}$, define the Prüfer variables $(r_u(n),\theta_u(n))\in\R_+^*\times[0,2\pi)$ such that
$U(n)=r_u(n)\begin{pmatrix}
\sin \theta_u(n)\\
\cos \theta_u(n)
\end{pmatrix}$ 	and define $t_u:=\text{sgn}(\tan \theta_u) \inf\left(|\tan \theta_u|,\dfrac{1}{|\tan \theta_u|}\right)$ and the same for $t_v$. Using these notations \eqref{gradcoli} can be rewritten
\begin{equation}
\|r_u-r_v\|_1\leq e^{-l^\beta/2}.
\end{equation}

Now for $n\in\N$ and $m\in\llbracket 0,N-1\rrbracket$ let 
\begin{equation}\label{defmatP}
P_{n,F,m}=\dfrac{1}{b_{Nn+m+1}}\begin{pmatrix}
F-a_m\omega_n & -b_{Nn+m}\\
b_{Nn+m+1} & 0
\end{pmatrix}
\end{equation}
 be the 1-step transfer matrix for \eqref{fullopdis} that goes from $\{Nn+m,Nn+m-1\}$ to $\{Nn+m+1,Nn+m\}$.

As in the previous subsection, in the following lines $\varepsilon$ will denote a vector such that $\|\varepsilon\|\leq Ce^{-l^\beta/2}$, its value may change from a line to another. For all $n\in\llbracket 1,\dots,l\rrbracket$ we have 
\begin{displaymath}
M_{n+1}^F U(n+1)=A_n^F N_n^F U(n)+\varepsilon
\end{displaymath}
where $M_n^F:=\begin{pmatrix}
e_{1,F}^n(1) & e_{2,F}^n(1)\\
e_{1,F}^n(0) & e_{2,F}^n(0)
\end{pmatrix}$, 
$N_n^j:=\begin{pmatrix}
e_{1,F}^n(N-1) & e_{2,F}^n(N-1)\\
e_{1,F}^n(N-2) & e_{2,F}^n(N-2)
\end{pmatrix}$ and
\\
$A_n^F:=P_{n+1,F,0}P_{n,F,N-1}$. \\
Thus, if we define $T_{n,F}^+:= (M_{n+1}^F)^{-1}A_n^FN_n^F$ we have
\begin{displaymath}
U(n+1)=T_{n,F}^+\,U(n)+\varepsilon \text{  and  } V(n+1)=T_{n,G}^+\,V(n)+\varepsilon.
\end{displaymath}
Now, as in Subsection~\ref{subseccont}, we prove the 
 \begin{lem}\label{closetan2}
There exist nine analytic functions $(f_i)_{i\in\llbracket 0,8\rrbracket}$ (only depending on $q$ and $N$) defined on $\R^3$ and not all constantly equal to zero such that, if $u$ (respectively $v$) is a 1-normalized eigenfunction of $H_\omega$ associated to $E_j(\omega)\in\left[F-e^{-l^\beta},F+e^{-l^\beta}\right]$ (respectively associated to $E_k(\omega)\in[G-e^{-l^\beta},G+e^{-l^\beta}]$), if for some $n_0\in\Z\cap(-l+8N,l-8N)$ we have $r_u(0)\geq e^{-l^\beta/4}$ and
\begin{displaymath}
\forall m\in\{n_0-1,n_0,n_0+1\}, \left|r_u(m)-r_v(m)\right|\leq e^{-l^\beta/2} 
\end{displaymath}
and if we define the polynomials 
\begin{align*}
\mathcal{R}_\omega(X)&:=\sum_{i=0}^8 f_i\left(\omega_{n_0-1},\omega_{n_0},\omega_{n_0+1}\right) X^i\\
\mathcal{Q}_\omega(X)&:=\sum_{i=0}^8 f_{8-i}\left(\omega_{n_0-1},\omega_{n_0},\omega_{n_0+1}\right) X^i
\end{align*}
then we have : 
\begin{align*}
\text{if }\exists\, g\in\{\tan,\cot\}, 
\left\{\begin{aligned}
t_v(n_0)=g(\theta_v(n_0))\\
t_u(n_0)=g(\theta_u(n_0))
\end{aligned}\right., 
\text{ then }&\left|\mathcal{R}_{\hat{\omega}}\left(t_v\left(n_0\right)\right)\right|\leq e^{-l^\beta/4},\\
\text{otherwise, we have }  &\left|\mathcal{Q}_{\hat{\omega}}\left(t_v\left(n_0\right)\right)\right|\leq e^{-l^\beta/4}.
\end{align*}
\end{lem}
\begin{proof}
The proof follows the one of Lemma~\ref{closetan} and we will use its notation. As in the proof of Lemma~\ref{closetan}, it suffices to show that there exists $(\omega_{n-1},\omega_n,\omega_{n+1})$ such that the matrix
 \begin{equation}
A:=\begin{pmatrix}
\Delta_1 & 0 & \Delta_3 & 0\\
\Pi_+ & \Delta_1 & \Pi_- & \Delta_3\\
\Delta_2 & \Pi_+ & \Delta_4 & \Pi_-\\
0 & \Delta_2 & 0 & \Delta_4
\end{pmatrix}
\end{equation}
where we have defined 
\begin{displaymath}
\left\{\begin{aligned}
\Delta_1&:=\|T_{n,F}^+ (1,0)\|^2_a -\|T_{n,G}^+ (1,0)\|^2_a,\\
\Pi_+&:=\left\langle T_{n,F}^+ (1,0),T_{n,F}^+ (0,1)\right\rangle_a,\\
\Delta_2&:=\|T_{n,F}^+ (0,1)\|^2_a -\|T_{n,G}^+ (1,0)\|^2_q,\\
\Delta_3&:=\|T_{n,F}^- (1,0)\|^2_a -\|T_{n,G}^- (1,0)\|^2_q\\
\Pi_-&:=\left\langle T_{n,F}^- (1,0),T_{n,F}^- (0,1)\right\rangle_a,\\
\Delta_4&:=\|T_{n,F}^- (0,1)\|^2_a -\|T_{n,G}^- (1,0)\|^2_q.
\end{aligned}
\right. 
\end{displaymath}
has its determinant not equal to zero. This will be done under the assumption 
\begin{equation}\label{hyp}
\omega_{n+1}=\omega_{n}=\omega_{n-1} ,
\end{equation} 
which will be supposed from now on. We now compute an equivalent of the determinant when $\omega_n$ tends to infinity.

If $N\geq 3$ we have,
$\begin{pmatrix}
e_{1,F}^n(N-1)  & e_{2,F}^n(N-1)\\
e_{1,F}^n(N-2)  & e_{2,F}^n(N-2)
\end{pmatrix}=\begin{pmatrix}
\frac{1}{\sqrt{a_{N-1}}}  & 0  \\
0                 & 0
\end{pmatrix}+o(1)
$ and
$\begin{pmatrix}
e_{1,F}^n(1)  & e_{2,F}^n(1)\\
e_{1,F}^n(0)  & e_{2,F}^n(0)
\end{pmatrix}=\begin{pmatrix}
0 & 0   \\
0 & \frac{1}{\sqrt{a_0}}
\end{pmatrix}+o(1)$. If $N=2$ we have 
\begin{displaymath}
\begin{pmatrix}
e_{1,F}^n(N-1)  & e_{2,F}^n(N-1)\\
e_{1,F}^n(N-2)  & e_{2,F}^n(N-2)
\end{pmatrix}=
\begin{pmatrix}
e_{1,F}^n(1)  & e_{2,F}^n(1)\\
e_{1,F}^n(0)  & e_{2,F}^n(0)
\end{pmatrix}=
\begin{pmatrix}
\frac{1}{\sqrt{a_1}}  & 0  \\
0   & \frac{1}{\sqrt{a_0}}
\end{pmatrix}.
\end{displaymath}
Now, $\|T_{n,F}^+ (1,0)\|_a$ is the $a$-norm of the vector $y(.)$ satisfying \eqref{plan} and such that $\begin{pmatrix}
y(1)\\
y(0)
\end{pmatrix}
=A_n^F 
\begin{pmatrix}
e_{1,F}^n(N-1)  \\
e_{1,F}^n(N-2)  
\end{pmatrix}.$

In the following lines, we will keep the difference $(\omega_{n+1},\omega_{n},\omega_{n-1})$ in notations, although they are equal, for a better comprehension of all terms. When $\omega_{n}\to\infty$, we have
\begin{displaymath}
\begin{pmatrix}
y(1)\\
y(0)
\end{pmatrix}=\dfrac{1}{b_{N(n+1)+1}b_{N(n+1)}}\begin{pmatrix}
\dfrac{(F-a_{N-1}\omega_n)(F-a_{0}\omega_{n+1})}{\sqrt{a_{N-1}}}\\
\dfrac{b_{N(n+1)+1}(F-a_{N-1}\omega_n)}{\sqrt{a_{N-1}}}
\end{pmatrix}
+o(\omega_n).
\end{displaymath}
Hence, for all $m\in\llbracket1,N-2\rrbracket$ we have
\begin{displaymath}
\begin{pmatrix}
y(m+1)\\
y(m)
\end{pmatrix}=\left(\sqrt{a_{N-1}}\overset{m+1}{\underset{i=0}{\prod}} b_{N(n+1)+i}\right)^{-1}Y
+o(\omega_n^m)
\end{displaymath}
where
\begin{displaymath}
Y=\begin{pmatrix}
(F-a_{N-1}\omega_n)\overset{m}{\underset{i=0}{\prod}}\left(F-a_{i}\omega_{n+1}\right)\\
\\
b_{N(n+1)+m+1}(F-a_{N-1}\omega_n)\overset{m-1}{\underset{i=0}{\prod}}\left(F-a_{i}\omega_{n+1}\right)
\end{pmatrix}.
\end{displaymath}
Therefore, we obtain
\begin{displaymath}
\left\{\begin{aligned}
&\left\|T_{n,F}^+ (1,0)\right\|^2_a=o(\omega_n^{2N-3})\\
&~~+(F-a_{N-1}\omega_n)^2\overset{N-3}{\underset{i=0}{\prod}}\left(F-a_{i}\omega_{n+1}\right)^2\dfrac{(F-a_{N-2}\omega_{n+1})^2+b_{N(n+2)-1}^2}{\mathcal{B}_+^2}\\
&\left\|T_{n,F}^- (0,1)\right\|^2_a=o(\omega_n^{2N-3})\\
&~~+(F-a_{0}\omega_n)^2\overset{N-1}{\underset{i=2}{\prod}}\left(F-a_{i}\omega_{n-1}\right)^2\dfrac{(F-a_{1}\omega_{n+1})^2+b_{N(n-1)}^2}{\mathcal{B}_-^2}
\end{aligned}\right. 
\end{displaymath}
where $\mathcal{B}_+:=\overset{N-1}{\underset{i=0}{\prod}} b_{N(n+1)+i}$ and 
$\mathcal{B}_-:=\overset{N-1}{\underset{i=0}{\prod}} b_{N(n-1)+i}$.
 In the same way we obtain $\left\|T_{n,F}^- (1,0)\right\|^2_a=o(\omega_n^{2N-1})$.

 Now, define $\mathcal{A}:=\overset{N-1}{\underset{i=0}{\prod}} a_i$ and $\mathcal{A}_0:=\overset{N-1}{\underset{i=0}{\sum}} \underset{m\neq i}{\prod} a_m$. When $\omega_{n}\to\infty$, we compute
\begin{equation}
\left\{\begin{aligned}
\Delta_1&=\dfrac{1}{\mathcal{B}_+^2}\mathcal{A}\mathcal{A}_0(F-G) \omega_n^{2N-1}+o(\omega_n^{2N-1})\\
\Pi_+&=O\left(\omega_n^{2N-1}\right)\\
\Delta_2&=O\left(\omega_n^{2N}\right)\\
\Delta_3&=o\left(\omega_n^{2N-1}\right)\\
\Pi_-&=O\left(\omega_n^{2N-1}\right)\\
\Delta_4&=\left(\dfrac{\mathcal{A}}{\mathcal{B}_-}\right)^2\omega_n^{2N}+o(\omega_n^{2N-1})
\end{aligned}\right.~~~.
\end{equation}

Now, a simple calculation shows that
\begin{equation}
\det A=(\Pi_+\Delta_3-\Pi_-\Delta_1)(\Pi_+\Delta_4-\Pi_-\Delta_2)+\left(\Delta_1\Delta_4-\Delta_2\Delta_3\right)^2.
\end{equation}
Hence, when $\omega_n\to\infty$,
\begin{displaymath}
\det A \sim (\Delta_1\Delta_4)^2\sim (F-G)^2\mathcal{A}_0^2\mathcal{A}^3\left(\dfrac{1}{\mathcal{B}_+\mathcal{B}_-}\right)^4\omega_n^{8N-2}.
\end{displaymath}
Therefore, $\det A$ is not constantly zero. We conclude the proof of Lemma~\ref{closetan2} in the same way we conclude the proof of Lemma~\ref{closetan}, using Theorem~\ref{Loja}.
\end{proof}

Let $T_{n,G}^0:=\overset{N-1}{\underset{m=0}{\prod}} P_{n,G,m}$ be the $N$-step transfer matrix for \eqref{fullopdis} that goes from $\{Nn-1,Nn\}$ to $\{Nn+N,Nn+N+1\}$. The matrix $T_{n,G}^0$ is the product of all the matrices $P_{n,J,m}$ that depends on $\omega_n$. Now, define
\begin{displaymath}
\begin{aligned}
\mathcal{T}_{n,G}^0:\mathbb{T}&\to\mathbb{T}\\
\theta&\to \Upsilon^{-1}\left[\dfrac{T_{n,G}^0\begin{pmatrix}
\sin(\theta)\\
\cos(\theta)
\end{pmatrix}}{\left\|T_{n,G}^0\begin{pmatrix}
\sin(\theta)\\
\cos(\theta)
\end{pmatrix} \right\|}
\right]
\end{aligned}
\end{displaymath} 
where we recall that $\mathbb{T}:=\R/\pi\Z$ is endowed with the distance $d$ and $\Upsilon:\mathbb{T}\to S^1/{-1,1}$ is bi-Lipschitz. The number $\mathcal{T}_{n,G}^0(\theta)$ is the direction of $T_{n,G}^0X$ where $X\in\R^2$ is any non-zero vector of direction any representative of $\theta$.
  We now prove the
\begin{lem}\label{transtan2}
Fix $(\Theta_1,\Theta_2)\in \mathbb{T}^2$ and let $\mathbb{P}_\epsilon$ denote the probability that \\
$d\big(\mathcal{T}_{n,G}^0(\Theta_1),\Theta_2\big)\leq \epsilon$. Then, there exist $\epsilon_0>0$ and $M\in\N$ such that for all $0<\epsilon<\epsilon_0$, $\mathbb{P}_\epsilon\leq \left(\dfrac{\epsilon}{\epsilon_0}\right)^{1/M}$.
\end{lem}
\begin{proof}
Using Proposition~\ref{Loja2}, as in the proof of Lemma~\ref{transtan}, it suffices to show that, for all $\theta_1$, the non zero vector $Z:=T_{n,G}^0\begin{pmatrix}
\sin(\theta_1)\\
\cos(\theta_1)
\end{pmatrix}$, which is analytic in $\omega_n$, do not have fixed direction.

We know from \eqref{defmatP}  that
\begin{displaymath}
T_{n,G}^0=\begin{pmatrix}
P_1(\omega_n) & P_2(\omega_n)\\
P_3(\omega_n) & P_4(\omega_n)
\end{pmatrix}
\end{displaymath}
 where $(P_i)_i$ are polynomials of $\omega_n$. In the limit $\omega_n\to\infty$, we compute
\begin{equation}\label{simPol}
\left\{\begin{aligned}
 P_1(\omega_n)&\sim \omega_n^{N}\prod_{m=0}^{N-1} \dfrac{(-a_m)}{b_{nN+m+1}}\\
 P_3(\omega_n)&=O(\omega_n^{N-1})
\end{aligned}\right.
\end{equation}
Now, we know from \eqref{defmatP} that $P_{n,G,1}\begin{pmatrix}
0\\
1
\end{pmatrix}=-b_{Nn+1}\begin{pmatrix}
1\\
0
\end{pmatrix}$. We then compute, in the same manner as \eqref{simPol}, when $\omega_n\to\infty$, 
\begin{equation}\label{simPol2}
\left\{\begin{aligned}
 P_2(\omega_n)&\sim -\omega_n^{N-1}b_{Nn+1}\overset{N-1}{\underset{m=1}{\prod}} \dfrac{(-a_m)}{b_{nN+m+1}}\\
 P_4(\omega_n)&=O(\omega_n^{N-2})
\end{aligned}\right.
\end{equation}
Besides, as $\det P_{n,G,m}=-\dfrac{b_{nN+m}}{b_{nN+m+1}}$, there exists $K_n\in\R^*$ such that, for all $\omega_n$, we have $\det T_{n,G}^0=K_n$. Thus, we obtain the equation 
\begin{displaymath}
P_1\cdot P_4 - P_2\cdot P_3=K_n.
\end{displaymath}
with $\deg(P_1)=N=\deg(P_2)+1\geq 2$, $\deg(P_3)\leq N-1$ and $\deg(P_4)\leq N-2$. This shows in particular that $\deg(P_3)>\deg(P_4)$, hence $P_3$ is not the zero polynomial. Therefore, $P_4$ cannot be the zero polynomial. Indeed, if it were we would then have $P_2\cdot P_3=K_n$.

Now, fix $\theta_1\in\mathbb{T}$. By definition, we have
\begin{displaymath}
Z=\begin{pmatrix}
\sin(\theta_1) P_1(\omega)+\cos(\theta_1) P_2(\omega_n)\\
\sin(\theta_1) P_3(\omega)+\cos(\theta_1) P_4(\omega_n)
\end{pmatrix}.
\end{displaymath}
Using \eqref{simPol} when $\theta_1\neq \pi\Z$ and \eqref{simPol2} when $\theta=\pi\Z$, the direction of the vector $Z$ get close to the axis directed by $(1,0)$. Therefore, it suffices to show that $Z$ is not constantly co-linear to $(1,0)$ and conclude. As $\deg(P_3)>\deg(P_4)\leq 0$, the second component of the vector $Z$ cannot be constantly equal to zero for any value of $\theta$. This concludes the proof of Lemma~\ref{transtan2}.
\end{proof}

Now, Lemma~\ref{probcoli} is proven in the same way as in Subsection~\ref{subseccont} using Lemma~\ref{closetan2} and Lemma~\ref{transtan2} instead of Lemma~\ref{closetan} and Lemma~\ref{transtan}.

\appendix
\section{Analytic functions of several real variables}
In this section, we prove two properties of analytic functions of several real variables. For $x\in\R^n$, we will write $x=(\hat{x},x_n)$. We will use the Weierstrass preparation theorem (see for instance \cite{L65}), which we recall now. 
\begin{theo}\label{Weierstrass}
Let $\mathcal{O}$ be an open subset of $\R^n$ that contains the origin and let $f:\mathcal{O}\to\R$ be an analytic function vanishing at the origin such that the analytic function $x_n\mapsto f(0,\cdots,0,x_n)$ has a zero of order $m\in\N^*$ at $0$. There exists a neighborhood $\mathcal{U}$ of the origin, a Weierstrass polynomial $P(\hat{x},x_n)=a_0(\hat{x})+ a_1(\hat{x})x_n+\cdots+a_{m-1}(\hat{x})x_n^{m-1}+x_n^m$, defined on $\mathcal{U}$, with $a_i(0)=0$ for all $i\in\llbracket 1,m-1\rrbracket$, and an analytic function $g:\mathcal{U}\to\R$ with $g(0)\neq 0$, such that, for all $x\in\mathcal{U}$, we have $f(x)=P(x)g(x)$.
\end{theo}

We now prove the
\begin{prop}\label{Loja}
Fix $\Omega\subset \R^n$ an open set and $f:\Omega\to \R$ a non zero analytic function. Fix $G$ a compact subset of $\Omega$. There exist $\epsilon_0>0$ and $m\in\N^*$ such that, for all $0<\epsilon<\epsilon_0$, we have $\left|\{x\in G, |f(x)|<\epsilon\}\right |\leq  \left(\frac{\epsilon}{\epsilon_0}\right)^{1/m}$.
\end{prop}
\begin{proof}
Since $G$ is compact, it suffices to prove that there exists a neighborhood of every point on which the result holds. Therefore, we fix a point $x\in G$. If $f(x)\neq 0$, the result is clear; now suppose $f(x)=0$. Without loss of generality, we can suppose that $x=0$. As $f$ is not constantly equal to zero, we can suppose that $h_n\mapsto f(0,\dots,0,h_n)$ is not constantly equal to zero near in a neighborhood of $0$ (if this is not the case, we do a rotation on the arguments) and therefore has a zero of order $m$ at $0$ for some $m\in\N^*$. The Theorem~\ref{Weierstrass} shows that there exist a neighborhood $\mathcal{U}$ of the origin, a Weierstrass polynomial $P$ as in Theorem~\ref{Weierstrass} and an analytic function $g$ not vanishing at the origin, such that, for all $x\in\mathcal{U}$, we have $f(x)=P(x)g(x)$. Now, take $\mathcal{V}:=]-\delta,\delta[^n$ included in $\mathcal{U}$ such that $\underset{x\in\mathcal{V}}{\inf} \,g(x)=:\frac{1}{C}>0$. Thus, if $x\in\mathcal{V}$ satisfies $|f(x)|\leq \epsilon$, then $|P(x)|\leq C\epsilon$. Therefore, it remains to prove the result for the Weierstrass polynomial $P$.

Fix $\epsilon>0$, $\hat{x}\in]-\delta,\delta[^{n-1}$ and suppose $|P(\hat{x},x_n)|\leq\epsilon$. As the polynomial $X_n\mapsto P(\hat{x},X_n)$ is unitary, let $(z_i(\hat{x}))_{i\in\llbracket 1,m\rrbracket}$ be its complex roots, repeated according to multiplicity. As $|P(x)|\leq\epsilon$, there exists $i\in\llbracket 1,m\rrbracket$ such that $|x_n-z_i(\hat{x})|\leq \epsilon^{\frac{1}{m}}$. Furthermore, $x_n$ is real number, hence we have $|x_n-\Re\left(z_i(\hat{x})\right)|\leq \epsilon^{\frac{1}{m}}$ and
\begin{displaymath}
x_n\in \mathcal{A}_\epsilon(\hat{x}):=\bigcup_{i\in\llbracket 1,m\llbracket} \left]\Re\left(z_i(\hat{x})\right)-\epsilon^{\frac{1}{m}},\Re\left(z_i(\hat{x})\right)+\epsilon^{\frac{1}{m}}\right[\,.
\end{displaymath}
The size of $\mathcal{A}_\epsilon(\hat{x})$ is smaller than $2m\epsilon^{\frac{1}{m}}$. Therefore, we have
\begin{equation*}
\left|\{x\in \mathcal{V}, |P(x)|< \epsilon\}\right|\leq \int_{\hat{x}\in]-\delta,\delta[^{n-1}} \left(\int_{\mathcal{A}_\epsilon(\hat{x})} dx_n\right) d\hat{x} \leq (2\delta)^{n-1}2m\epsilon^{\frac{1}{m}}.
\end{equation*}
This conclude the proof of Proposition~\ref{Loja}.
\end{proof}

We now prove an other proposition, connected to Proposition~\ref{Loja}.

\begin{prop}\label{Loja2}
Fix $\Omega\subset \R^n$ an open set containing the origin and $f:\Omega\to \R$ an analytic function such that, for all $(\hat{x},x_n)\in\Omega$, the function $h_n\mapsto f(\hat{x},x_n+h_n)$ is not constantly equal to zero in a neighborhood of the $0$. Fix $G:=[-M,M]^n$ a compact subset of $\Omega$. There exists $\epsilon_0>0$ and $m\in\N^*$ such that, for all $0<\epsilon<\epsilon_0$ and $\hat{x}\in[-M,M]^{n-1}$, we have $\left|\{x_n\in[-M,M], |f(\hat{x},x_n)|<\epsilon\}\right |\leq  \left(\frac{\epsilon}{\epsilon_0}\right)^{1/m}$.
\end{prop}
\begin{proof}
As $G$ is compact, it suffices to proce that, for all $(\hat{x},x_n)\in[-M,M]^n$, there exist $\delta>0$, $\epsilon_0>0$ and $m\in\N^*$ such that, for all $\hat{y}\in\hat{x}+]-\delta,\delta[^{n-1}$, we have 
\begin{displaymath}
\left|\{h_n\in [-\delta,\delta], |f(\hat{y},x_n+h_n)|<\epsilon\}\right |\leq  \left(\frac{\epsilon}{\epsilon_0}\right)^{1/m}.
\end{displaymath}
As the function $h_n\mapsto f(\hat{y},x_n+h_n)$ is not constantly equal to zero in a neighborhood of $0$, we can apply the Theorem~\ref{Weierstrass} according to the $n$th coordinate and conclude as in the proof of Proposition~\ref{Loja}.
\end{proof}

\section{Non oscillating solutions of Sturm-Liouville equations}\label{sturm}
In this section we prove two lemmas that are used in subsection~\ref{subseccont}. Let $q$ be a positive continuous function on $(0,1)$, $\tilde{q}$ be the one-periodic function that is equal to $q$ on $(0,1)$ and $E_2<E_1$. Suppose $m:=\underset{x(0,1)}\inf q(x)>0$. We define the ODEs
\begin{equation}\label{ODEapp}
(\mathcal{E}_\lambda^j) : \forall x\in\R,y''(x)=\left(\lambda \tilde{q}(x)-E_j\right) y(x).
\end{equation}

The purpose of this section is to prove that for $\lambda$ large enough, the solutions of \eqref{ODEapp} don't oscillate and therefore can be easily compared one another.
\begin{lem}\label{onezero}
Fix $\lambda$ such that $\lambda m > E_j$. Then, any non-zero solution of $(\mathcal{E}_\lambda^j)$ has at most one zero.
\end{lem}
\begin{proof}
Let $u$ be a non-zero solution of $(\mathcal{E}_\lambda^j)$ with a zero at point $x_0$. Without loss of generality we can suppose that $u'(x_0)>0$. Then, for any $x>x_0$, $u(x)>0$. If not, there would exists $x_1$ such that $u''(x_1)<0$ and $u(x_1)>0$, which is in contradiction with the fact that $\lambda m>E_j$.
In the same way, we prove that for any $x<x_0$, $u(x)<0$.
\end{proof}

\begin{lem}\label{threezero}
Fix $\lambda$ such that $\lambda m > E_1>E_2$. Let $u$ (respectively $v$) be a non-zero solution of $(\mathcal{E}_\lambda^1)$, positive near $+\infty$ (respectively $(\mathcal{E}_\lambda^2)$, positive near $+\infty$). Then $v-u$ changes of sign at most three times.
\end{lem}
\begin{proof}
Indeed, if $w:=v-u$, then $w$ satisfies the equation 
\begin{equation}\label{ODEdiff}
w''=(\lambda \tilde{q}-E_1)w+(E_1-E_2)v.
\end{equation}
Now, suppose there exists $x_+$ such that $v(x_+)\geq 0$ and such that $w(x_+)=0$ and $w$ is negative on the left side, positive on the right side of $x_+$. Then, for $x>x_0$, $w(x)>0$ and $v(x)>0$. If not, there would exists $x_1>x_0$ such that $v(x_1)>0$, $w(x_1)>0$ and $w''(x_1)\leq 0$, which would be in contradiction with \eqref{ODEdiff}. 

In the same way we prove that if there exists $x_-$ such that $v(x_-)\leq 0$, $w(x_-)=0$ and $w$ is negative on the left side, positive on the right side of $x_-$, then for $x<x_-$, $w(x)<0$ and $v(x)<0$. Now, as $w$ takes both sign in $(x_-,x_+)$ there exists one and only one $x_c\in(x_-,x_+)$ such that $w$ changes of sign at $x_c$.
\end{proof}

\begin{rem}\label{remzero}\normalfont
One can prove the results of this section with a potential $q$ that has isolated zero. Indeed, in an interval that contain only one zero of q, any solution of \eqref{ODEapp} oscillate with a pulsation close to $E$. But as $\lambda$ tends to infinity, the length of such an interval tends to zero. Thus, for $\lambda$ large enough, the function are still strictly increasing and one can prove Lemma~\ref{onezero} and Lemma~\ref{threezero} in the same way as above.
\end{rem}

\section{Non oscillating solutions of finite difference equations of order two}\label{FDEsec}
 There are discrete equivalent of the results of Appendix~\ref{sturm}.  Fix $k\in\N$ and $E_2>E_1$. Let $(a_n)_{n\in\llbracket 1,k\rrbracket}\in\R^k$ and $(\tilde{a}_m)_{m\in\Z}$ its $k$-periodic extension. We suppose there exists $m>0$ such that

\begin{equation}
\min(a_n)_{n\in\llbracket 1,k\rrbracket}>m.
\end{equation}
We define the finite-difference equation :
\begin{equation}\label{FDEapp}
(\mathcal{E}_\lambda^j) : \forall m\in\Z,u(m+1)+u(m-1)=\left(E_j-\lambda \tilde{a}_m\right) u(m).
\end{equation}

\begin{lem}\label{onezero2}
Fix $\lambda<0$ such that $E_j-\lambda m>2$. Then, any non-zero solution of $(\mathcal{E}_\lambda^j)$ changes of sign at most one time.
\end{lem}
\begin{proof}
Let $u$ be a non-zero solution of $(\mathcal{E}_\lambda^j)$ that changes sign at $(n_0,n_0+1)$. Without loss of generality we can suppose that $u(n_0+1)>u(n_0)$. Then, for any $n\geq n_0+1$, $u(n)>0$. If not, there would exists $n_1$ such that $u(n_1)>0$ and $u(n_1+1)+u(n_1-1)<2u(n_1)$. Therefore, we would have
\begin{displaymath}
u(n_1+1)+u(n_1-1)<2u(n_1)
\end{displaymath} which is in contradiction with the fact that $E_j-\lambda m>2$.
In the same way, we prove that for any $n\leq n_0$, $u(n_0)<0$.
\end{proof}

\begin{lem}\label{threezero2}
Fix $\lambda<0$ such that $E_2-\lambda m > E_1-\lambda m>2$. Let $u$ (respectively $v$) be a non-zero solution of $(\mathcal{E}_\lambda^1)$ (respectively $(\mathcal{E}_\lambda^2)$). Then $v-u$ changes of sign at most three times.
\end{lem}
\begin{proof}
Indeed, if $w:=v-u$, then $w$ satisfies the equation 
\begin{equation}\label{FDEdiff}
w(n+1)+w(n-1)=(E_1-\lambda a_n) w(n)+(E_2-E_1)v.
\end{equation}
Now, suppose there exists $n_+$ such that $v(n_+)\geq 0$, $w(n_+)\leq 0$ and \\
$w(n_+ +1)\geq 0$. Then, for $n\geq n_+ +1$, $w(n)>0$ and $v(n)>0$. If not, there would exists $n_1\geq n_++1$ such that $v(n_1)>0$, $w(n_1)>0$ and $w(n_1-1)+w(n_1+1)<2w(n_1)$, which would be in contradiction with \eqref{FDEdiff}.

In the same way we prove that if there exists $n_-$ such that $v(n_-)\leq 0$, $w(n_- -1)\leq 0$ and $w(n_-)\geq 0$ then for $n\leq n_- -1$, $w(n)<0$ and $v(n)<0$. Now, as $w$ takes both sign in $(n_-,n_+)$ there exists one and only one $n_c\in\llbracket n_-,n_+ \rrbracket$ such that $w$ changes of sign at $x_c$.
\end{proof}

\bibliographystyle{plain}
\bibliography{biblio2}

\end{document}